\documentclass[12pt,thmsa]{article}
\usepackage{amsthm,amssymb,amsfonts,graphicx,lscape}
\usepackage{booktabs}
\usepackage[T1]{fontenc}
\usepackage[sc]{mathpazo}
\usepackage{epstopdf}
\usepackage{subfig}
\usepackage{enumerate}
\usepackage{booktabs}
\usepackage{setspace}
\usepackage{amsmath}
\usepackage[mathscr]{euscript}
\usepackage[svgnames]{xcolor}
\usepackage{geometry}
\usepackage{rotating}
\usepackage{xargs}
\usepackage{ulem}
\usepackage[authoryear,round,longnamesfirst]{natbib}
\usepackage{hyperref}
\usepackage{csquotes}
\usepackage{sectsty}
\usepackage[capitalize,nameinlink]{cleveref}
\usepackage[font={it},justification=justified]{caption}
\usepackage{dcolumn}
\usepackage{comment}
\usepackage{makecell}
\hypersetup{
    pdftitle={Strategic Complexity Promotes Egalitarianism in Legislative Bargaining},
    pdfauthor={Marina Agranov, S. Nageeb Ali, B. Douglas Bernheim, and Thomas R. Palfrey},
    colorlinks,
    breaklinks,
    naturalnames,
    citecolor=DarkBlue,
    bookmarksnumbered=true,
    urlcolor=Indigo,linkcolor=DarkBlue
}
\renewcommand{\Re}{\mathbb R}
\renewcommand{\epsilon}{\varepsilon}
\newcommand\crossout{\bgroup\markoverwith
      {\textcolor{blue}{\rule[0.5ex]{2pt}{0.8pt}}}\ULon}

\DeclareMathOperator*{\argmax}{arg\,max}

\newtheorem{proposition}{Proposition}
\newtheorem{lemma}{Lemma}

\newtheorem{assumption}{Assumption}

\sectionfont{\color{DarkRed}}
\subsectionfont{\color{DarkRed}}
\subsubsectionfont{\color{DarkRed}}
\begin{document}

\title{Strategic Complexity Promotes Egalitarianism in Legislative Bargaining\thanks{We thank John Duffy, Michael McBride and the staff of the Experimental Social Science Laboratory (ESSL) at UC Irvine for their support and for granting access to the laboratory and subject pool. We thank the Hacker
Social Science Experimental Laboratory (SSEL) at Caltech for laboratory facilities and for supporting the software development for the experiment. The paper has benefited from comments by participants at various seminar presentations. Palfrey thanks the Department of Economics at Columbia University for their hospitality during his visit as a Mitchell Professor in Fall 2019. We thank Salvador Candelas, Han Ng, and Jeffrey Zeidel for their excellent research assistance. We gratefully acknowledge the financial support of the National Science Foundation (SES-1426560, SES-1530639, SES-2343948). The experiment was approved by the Caltech IRB (Protocol IR17-0271).}}

\author{Marina Agranov\thanks{Division of the Humanities and Social Sciences, California Institute of
Technology, Pasadena, CA} \and S. Nageeb Ali\thanks{Department of Economics, Pennsylvania State University, State College, PA} \and B. Douglas Bernheim\thanks{Department of Economics, Stanford University, Palo Alto, CA} \and Thomas R. Palfrey\thanks{Division of the Humanities and Social Sciences, California Institute of Technology, Pasadena, CA}}

\date{\today}
\maketitle
\thispagestyle{empty}

\begin{abstract}

Strategic models of legislative bargaining predict that proposers can extract high shares of economic surplus by identifying and exploiting weak coalition partners. However, strength and weakness can be difficult to assess even with relatively simple bargaining protocols. We evaluate experimentally how strategic complexity affects the ability to identify weak coalition partners, and for the partners themselves to determine whether their positions are weak or strong. We find that, as strategic complexity progressively obscures bargaining strength, proposers migrate to egalitarianism, in significant part because non-proposers begin placing substantial weight on fairness. Greater analytic skill dampens but does not eliminate these patterns.

\end{abstract}

\vspace{6mm}

 \vspace{6mm}

\onehalfspace
\clearpage

\setcounter{page}{1}

\section{Introduction}\label{Section-Introduction}

In legislative negotiations, legislators formulate successful proposals by making concessions in order to build winning coalitions. According to the applicable theoretical literature, two central principles govern this process. First, a skilled negotiator keeps the coalition as small as possible to minimize the number of concessions; second, the negotiator seeks out weak partners to minimize the sizes of those concessions.\footnote{These principles emerge in the seminal work of \cite{baron1989bargaining} and the subsequent theoretical literature on legislative bargaining.} However, it is also apparent from the pertinent literature that even relatively simple bargaining protocols make it difficult to assess which players are weak. In this paper, our broad objective is to evaluate the degree to which parties to such negotiations identify weak coalition partners and exploit perceived weakness. Our leading hypothesis is that strategic complexity leads people to make and support egalitarian proposals, even in settings where theory predicts that one party should be able to dictate the outcome.

To appreciate the strategic complexity of legislative bargaining, consider the nature of the task facing a proposer who seeks to craft the best deal for herself. The proposer anticipates that another player, if rational, would support a proposal only if its passage makes that player (weakly) better off than its rejection. But what happens following rejection may depend subtly on the dynamics of future play. To identify weak partners, a proposer must therefore understand how the process would subsequently play out. An additional wrinkle is that, even if the proposer can identify players who should, in theory, be willing to acquiesce for a small concession, she may not be entirely confident that weak players see themselves as weak. If weak players incorrectly believe that they are negotiating from positions of strength, proposals that attempt to exploit their objective weakness will fail. In such a scenario, a proposer might be inclined to propose a relatively egalitarian outcome to hedge against the possibility that she or others have incorrectly mapped out the future play.

The following example illustrates some of these subtleties. A three-person legislature consisting of Alice, Bob, and Carol is tasked with dividing one dollar. Bargaining takes place in a series of no more than three rounds. In each round, one legislator makes a proposal, which is then put to a vote. If any proposal receives two or more votes, it is adopted and the process ends. Alice makes the first proposal, Bob makes the second, and Carol makes the third. If all three proposals fail, the parties receive default payoffs (or outside options) $\mathbf{v}\equiv(v_{A},v_{B},v_{C})$. We assume that $v_A+v_B+v_C\leq 1$, and $v_i \geq 0$ for each $i$. Players do not discount future payoffs.

Which partner is weaker from Alice's perspective, Bob or Carol? One might conjecture that the answer hinges on the default payoffs: a lower value of $v_{B}$ (respectively $v_{C}$) makes Bob (respectively Carol) easier to exploit. And yet, barring a single knife-edge exception, in subgame perfect equilibria, Carol is \emph{always} the weaker partner for Alice, even if $v_{C}$ is large. In fact, Carol is so weak that she agrees to let Alice claim the entire dollar, which is the unique subgame perfect equilibrium outcome.\footnote{As is standard, we consider subgame perfect equilibria in which players do not play weakly dominated actions at the voting stage.} 

To understand this result, suppose that $v_A < v_B$, which implies $v_{A} < 0.5$. We analyze this game by backward induction. If the first two offers are rejected and the third round is reached, Carol would form a minimal winning coalition with Alice by proposing $(v_{A},0,1-v_{A})$, given that Bob's default payoff is higher. In the preceding round, Bob (the proposer) and the others expect to receive these continuation payoffs if his proposal is rejected. Because $v_A<(1-v_A)$, Bob will propose $(v_{A},1-v_{A},0)$, which will pass with Alice's support. Thus, Carol should expect to obtain zero if the first proposal is rejected and bargaining reaches the second round. Accordingly, in the first round, Alice will recognize and exploit Carol's extreme weakness and form a minimal winning coalition with her by proposing $(1,0,0)$.\footnote{We construe this proposal as a coalition with Carol because, in equilibrium, Carol resolves her indifference in favor of it. There do not exist equilibria in which Carol rejects with positive probability because she has a strict incentive to accept any proposal that offers her $\epsilon>0$.} Albeit with different dynamics, the same conclusion follows if $v_A > v_B$.\footnote{If the third round is reached, Carol will form a minimal winning coalition with Bob by proposing $(0,v_{B},1-v_{B})$. Consequently, if the second round is reached, Bob will form a minimal winning coalition with Alice by proposing $(0,1,0)$. Accordingly, in the first round, Alice will form a minimal winning coalition with Carol by proposing $(1,0,0)$. The knife-edge exception referenced in the text occurs when $\mathbf{v}=(0.5,0.5,0)$. In that case, there is an equilibrium in which Carol's proposal is $(0.5,0,0.5)$, Bob's proposal is $(0,0.5,0.5)$, and Alice's initial proposal is either $(0.5,0,0.5)$ or $(0.5,0.5,0)$.} Intuitively, Carol's strong position in the third round induces Bob to exclude her, which in turn makes her extremely vulnerable from Alice's first-round perspective.

One might worry that these theoretical predictions require a high degree of strategic sophistication. For instance, if $v_{B} < v_{C}$, Alice might mistakenly infer that Bob is the weaker partner. Alternatively, even if Alice is sophisticated, she might fear that Carol and Bob are naive. If Carol falsely believes that she will benefit in the continuation game because her default payoff is high, she will refuse offers that she ought to accept. Once Alice is uncertain about her own strategic reasoning, or about Bob's and Carol's, she may hedge her bets by treating a coalition partner more generously, or even by making an offer that benefits everyone (i.e., the grand coalition). Thus, even if no player inherently favors egalitarian outcomes, strategic complexity may promote egalitarianism as a hedge against uncertainty concerning the quality of others' or one's own reasoning about the game. 

We formalize our conjecture concerning egalitarianism by studying optimal offers in settings where the proposer does not know how others evaluate the continuation game that would follow rejection of her offer. We also assume that strategic complexity causes the proposer to either narrowly frame the problem by focusing only on the current stage game or hold pessimistic beliefs about her own continuation payoff. We show that, as the proposer's beliefs about others' strategic inferences become increasingly diffuse (in the sense of entropy) and their attention to or hopes for the continuation game dwindle, the optimal offer converges to a minimal winning coalition offer with an equal split. 

We evaluate these issues, including the aforementioned conjecture, by studying legislative bargaining protocols in a laboratory experiment. We have three main objectives: (1) determine the extent to which people identify and exploit weak coalition partners; (2) ascertain the reasons for failures to do so in settings where they occur; and (3) identify the systematic behavioral patterns (potentially including egalitarianism) that arise in those settings.   

With respect to the second objective, we decompose the main theoretical hypothesis into several components. The \emph{exploitative intent hypothesis} holds that proposers wish to extract as much surplus as possible from bargaining partners. The \emph{self-awareness hypothesis} holds that non-proposers understand their own weaknesses or strengths as coalition partners. The \emph{exploitative ability hypothesis} holds that proposers have the ability to exploit weak coalition partners, in terms of both identifying them and having a sufficiently nuanced understanding of continuation play that they can fine-tune their offers for maximal exploitation. If bargaining theory's predictions concerning the exploitation of weakness prove to be inaccurate, the fault could lie in any one of these component hypotheses. In particular, the exploitative intent hypothesis could fail if social motives overwhelm personal interest (e.g., if people are sufficiently fair-minded or altruistic), and the self-awareness and exploitative ability hypotheses could fail if people have limited ability to reason strategically (e.g., if they struggle with backward induction). 

 In the laboratory, we study three-player bargaining games with majority rule so that the ideal minimum winning coalition includes a single weak partner. Our experimental design involves treatments that differ with respect to the \textit{complexity} of the environment, the \emph{predictability} of bargaining power in future rounds, and the \emph{subject pool}. Because strategic reasoning is more complicated in games with larger numbers of bargaining rounds, we manipulate complexity by varying the length of the game from one round to three. We manipulate the predictability of bargaining power by varying the information subjects have about the identity of future proposers. According to the theoretical analysis of \cite{ali2019predictability}, greater predictability ought to allow initial proposers to obtain larger shares by weakening potential coalition partners. We vary the subject pool by recruiting participants at two universities: the University of California at Irvine (UCI) and the California Institute of Technology (Caltech). This variation is of interest because strategic reasoning draws on mathematical skills, and Caltech students score higher in mathematics on standardized tests.\footnote{According to PrepScholar, the average Math SAT scores for UCI and Caltech students are 675 and 795, respectively.} Finally, all of our analysis focuses on the behavior of subjects who have gained experience with the pertinent protocols.

We begin with an analysis of average payoffs. Across our treatments, the first proposers receive smaller shares than standard theory predicts. The discrepancy is particularly striking for three-round games, where these shares are below one-half of the total prize on average. Contrary to theory, the first proposer’s share in three-round games does not vary with the predictability of the proposer recognition order, which ought to affect the identifiability and exploitability of weak partners. Significantly, the first proposer’s share declines from one-round to two-round games, and from two-round to three-round games, despite the fact that every multiple-round game can be reduced to a single-round game (with equilibrium payments for the continuation game playing the role of the default payments). This pattern suggests that, in more complex bargaining environments, proposers face greater difficulty identifying and exploiting players who believe they are in weak positions. Moreover, as the bargaining game becomes more complex, outcomes migrate towards egalitarianism. Although social motives offer a contributory explanation for incomplete extraction of surplus by the first proposers, they cannot account for these comparative statics.

The remainder of our analysis examines proposals and voting behavior to determine how participants arrive at these bargaining outcomes.

First, we study the nature of proposals. We classify a proposal as \emph{dictatorial} if all non-proposing members would receive less than five percent, a \emph{minimal winning coalition} (henceforth MWC) if exactly one non-proposing member would receive at least five percent, and a \emph{grand coalition} if each non-proposing member would receive at least five percent. A proposal is \emph{egalitarian} within a coalition type if the difference in shares between coalition partners (including the proposer) is no more than five percent of the prize. In one-round games, Caltech subjects overwhelmingly make MWC offers (86\%), always to weak partners. Among those offers, only 9\% are egalitarian, and the average proposer share is 77\%. Thus, the proposers' behavior in these strategically simple settings shows clear indications of exploitative intent and exploitative ability. In two-round games, the fraction of MWC offers is about the same, but 12\% of those offers involve partners who are strong according to theory; these partners appear to be weak based on default payoffs. The average proposer share in MWC offers to weak partners dips slightly to 69\%, but egalitarianism remains relatively rare (7\%). These patterns point to a small decline in exploitative ability. In three-round games for which theory unambiguously identifies a weak partner, the fraction of MWC offers is 77\%, and 25\% of those involve theoretically strong partners. Moreover, the average proposer share for MWC offers to weak partners is just over 50\%, and 62\% of these proposals are egalitarian. All the other proposals in these games are grand coalition offers, two-thirds of which are egalitarian. We interpret these patterns as reflecting a diminished ability to identify weak partners as well as a dramatically reduced inclination to try to exploit them. 

To the extent exploitative intent is universal but exploitative ability is correlated with analytical skill, we might expect to see similar behavior at UCI in one-round games, along with a more immediate and dramatic shift toward egalitarianism as complexity rises. Our findings are generally consistent with that prediction. In one-round games, UCI subjects are significantly less likely to make MWC offers (68\%), but nearly all of them are to weak partners (99\%); of the MWC offers to weak partners, only 12\% are egalitarian, and the average proposer share is 70\%. Hence, we confirm exploitative intent, as well as a reasonably high degree of exploitative ability in this simple setting. In two-round games, the fraction of MWC offers increases to 74\% but nearly two-thirds of them are to partners who are strong according to theory; such players appear weak because of their smaller default payoffs. Among those offers, the egalitarian fraction rises to 63\%, and the average proposer share drops to 53\%. UCI subjects also exhibit no ability to identify weak partners in the three-round games for which theory unambiguously identifies a weak partner: with perfect predictability, only 54\% of the MWC offers include the weak partner, and with partial predictability, this fraction drops to 31\%. Moreover, among MWC offers to weak partners in the perfect-predictability games, egalitarianism runs as high as 90\% (with an average proposer payoff of 52\%). Among grand coalition offers in games with perfect predictability, egalitarianism rises from 18-19\% in one- and two-round games to 75\% in three-round games. Thus, we see egalitarianism emerging to a striking degree as a response to strategic complexity.

The next step in our analysis is to examine votes (that is, responses to proposals). We relate voting behavior to three variables: the player’s own share, whether the player is a weak or strong partner, and a measure of inequality. As expected, the likelihood of supporting a proposal always increases along with the player’s own share. 

We evaluate the self-awareness hypothesis by examining the relationship between voting and a player’s theoretical bargaining strength. Specifically, we ask whether strong partners are more difficult to recruit as coalition partners. In one-round games, behavior is consistent with this implication. The same is true for two-round games, but the strong-partner effect weakens in the UCI sample. In three-round games, the effect weakens in the Caltech sample and disappears entirely in the UCI sample. Thus, complexity reduces self-awareness of weakness, much as it reduces the ability of proposers to identify and exploit weakness. Interestingly, we find some evidence that the ability to identify weakness declines more rapidly with complexity for proposers than for non-proposers.\footnote{Specifically, in two-round games at UCI, non-proposers identify weakness more successfully than proposers.} A possible explanation is that non-proposers are more likely to think ahead to the future proposals they may make.

We provide further evidence concerning the hypothesis that complexity promotes egalitarianism by examining the relationship between voting and fairness. There is no indication that fairness matters to voters in one- and two-round games at Caltech, or in one-round games at UCI, but it does matter to them in three-round games at Caltech, and in two- and three-round games at UCI. These results imply that, to some degree, the observed relationship between the frequency of egalitarian proposals and strategic complexity is a rational response to voters' strategies. 

The final portion of our analysis evaluates the optimality of proposals given the voting behavior. We ask whether the migration towards egalitarian outcomes in more complex games is an optimization failure on the part of proposers or approximately optimal given the voting behavior. Our evidence supports the latter possibility: conditional on observed voting behavior, the optimal proposal migrates toward within-coalition egalitarianism as we increase the number of rounds. Consequently, actual offers to MWCs with weak partners are nearly optimal. Egalitarian offers to grand coalitions also become progressively less suboptimal as complexity rises. Accordingly, the nature of voting responses by non-proposing players accounts for much of the observed migration to egalitarianism as complexity rises.

The paper is organized as follows. \Cref{Section-Literature} relates our work to the pertinent literature. \Cref{Section-Games} explains the theoretical principles that guide our analysis and provides a conceptual foundation for the hypothesis that greater complexity leads to egalitarianism. \Cref{Section-Design} details the experiment. \Cref{Section-Results} presents our experimental findings. \Cref{Section-Conclusion} concludes. The appendices include tables supporting all bar graphs, additional tables and figures, and sample instructions.

\section{Related Literature}
\label{Section-Literature}

This paper draws on three important but largely separate areas of study: one on coalition formation and legislative bargaining, another on social preferences, and a third concerning the effects of complexity on decision making. The primary contribution of our paper is that it connects these ideas, offering a new rationale for fairness in legislative bargaining. 

The literature on legislative bargaining experiments, surveyed in \cite{palfrey2016experiments} and \cite{AgranovSurvey}, has generally focused on the degree to which proposers form and exploit minimal winning coalitions, whether institutional details such as the use of a closed-rule or open-rule bargaining protocol affects this behavior, and whether the comparative statics highlighted in \cite{baron1989bargaining} emerge in practice. Because most of the experimental treatments involve symmetrically positioned respondents and do not vary the complexity of the game, the questions at the core of this paper---namely, whether players can identify weak partners and how offers change with varying complexity---have received little attention. In this vein, some work studies three-player coalitional bargaining games with heterogeneous outside options. \cite{diermeier2006self} consider bargaining with a single-round, which resembles a three-player ultimatum game with simple majority rule, and find that the proposer has a slight tendency to target the player with the lowest disagreement value.\footnote{\cite{knez1995outside} were the first to study an ultimatum game with three players. Their experiment involved a proposer playing two separate ultimatum games in which the respondents had unequal outside options. The role of the third party was primarily to facilitate social comparisons.\ \cite{guth1998information} consider a setting in which one of two respondents has veto power, while the other is a dummy player with no ability to influence the outcome. Contrary to the predictions obtained from models with social preferences, they find that the proposer targets the respondent with veto power, and the dummy player receives a very small share.} \cite{miller2018legislative} study an environment with an indefinite horizon, and find that proposers are more likely to exclude the player with the higher disagreement payoff. These studies anticipate our finding concerning the validity of the exploitative intent hypothesis by showing that proposers take advantage of obviously weak players.

Several early experiments on two-person offer/counteroffer bargaining with a shrinking pie
present results that bear on some of the questions we ask here. Those studies interrogate the conditions under which negotiators act as "fairmen" or "gamesmen" and the extent to which such behavior depends on the length of the bargaining horizon and the rate of pie shrinkage. The results are mixed. In the first such study, \cite{binmore1985} compares one- and two-round bargaining games with a shrinking pie. The outcomes they observe in the more complex two-round game are \textit{less} egalitarian than in the original one-stage ultimatum game. First movers in the two-stage game are gamesmen, while first movers in the ultimatum game are fairmen. 

The \cite{binmore1985} study triggered a series of follow-ups that involved different variations on the shrinkage rate, the number of periods, and other design features. \cite{guthtietz1988} increased the shrinkage factor and raised the monetary stakes. Their results strongly rejected the game-theoretic solutions in all treatments and suggest that such theoretical failures will inevitably arise in situations where the equilibrium is socially unacceptable. \cite{neelin1988} compared behavior in two-, three-, and five-round bargaining games, with mixed results about the degree to which proposers follow the subgame perfect equilibrium strategy. The findings of \cite{spiegel1994} suggest that proposers try to achieve egalitarian outcomes when they are at a bargaining disadvantage. \cite{binmore1991} identify a potential role for simple rules of thumb and social norms, such as fairness, that bargainers might apply instinctively in some environments. However, the authors are agnostic about the conditions that trigger the application of such rules of thumb and cause proposers to abandon more strategic approaches. 

Because these studies focus on bilateral bargaining, negotiators do not need to identify  weak partners as in the more complex multi-person bargaining problems that we examine here. Our analysis highlights how, in the latter settings, people are willing to act as ``gamesmen'' when the bargaining protocol is easy to game, but switch to a hybrid ``fairmen/gamesmen'' strategy when it is complex, splitting the prize with a single participant while leaving the other with nothing.

Our work is also related to a vast literature on fairness and social preferences; see \cite{fehr-charness} for an exhaustive survey. In the settings we study, people do not exhibit much intrinsic concern for others' outcomes. In simple settings, they clearly exhibit exploitative intent, and even in complex settings, most proposals completely exclude one member of the group. Thus, it appears that our subjects are not intrinsically motivated either by fairness or, as in \cite{Andreoni2009}, by the desire to appear fair. Nevertheless, egalitarianism \emph{within} coalitions appears to emerge when strategic complexity makes it difficult to understand and anticipate the outcome of subsequent negotiations. 

Finally, our work contributes to the broad research agenda on how complexity influences decision making both in general and in strategic settings. Within economics, the study of behavioral responses to complexity traces its roots to Herbert Simon's seminal work on bounded rationality \citep{simon1956rational}. Over the last few years, interest in understanding how people make complex decisions has exploded  \citep[e.g.,][]{bernheim2020empirical,enke2023cognitive,oprea2024decisions,babohrenimas,puri2025simplicity}; see \cite{OpreaSurvey} for a survey. The inherent complexity of strategic reasoning has inspired the formulation of theories that seek to capture essential aspects of boundedly rational play, such as level-k reasoning \citep[e.g.,][]{nagel1995unraveling,stahl1995players,crawford2013survey}, cognitive hierarchies \citep[e.g.,][]{camerer2004cognitive,alaoui2016endogenous}, and finite automata  \citep[e.g.,][]{neyman1985bounded,Rubinstein1988,Kalai1990,Chatterjee2009}, as well as the search for strategically-simple mechanisms (surveyed in \citealt{li2024designing}). There is also evidence that people favor simple strategies in the infinitely repeated prisoner's dilemmas, and that simplicity can make particular strategies focal \citep{guanoprea2024faces}. Our measure of complexity---the number of rounds in a bargaining game---is ordinally equivalent to the number of subgames, and is therefore consistent with the measure used in work on finite automata. Our findings complement the complexity agenda by identifying a new effect of strategic complexity on behavior, namely that it promotes egalitarianism in bargaining.

\section{Theoretical Principles}\label{Section-Games}

In this section, we describe the theoretical principles that guide our experimental inquiry. \Cref{Subsection-equilibrium} describes equilibrium predictions. \Cref{Subsection-egalitarianism} explains why strategic complexity may lead to egalitarianism.

\subsection{Equilibrium predictions}\label{Subsection-equilibrium}

Our analysis concerns multi-party finite-round bargaining protocols. We focus on settings with three parties but the key principles discussed in this subsection generalize. In each round $t \in \{1,...,T\}$, one of the three players, $p_t \in \{A,B,C\}$, proposes a division of a fixed prize which, without loss of generality, we take to be one dollar. We write a proposal as $(x^A,x^B,x^C)$, where $\sum_{i \in \{A,B,C\}} x^i = 1$. All players then vote in favor or against this proposal. A proposal passes if a simple majority votes in favor, in which case negotiations end. If no proposal passes by the end of round $T$, the parties receive default payoffs $(v^A_T,v^B_T,v^C_T)$, where $\sum_{i \in \{A,B,C\}} v^i_T \leq 1$. We assume each player cares only about their own share of the prize. We also assume no discounting and risk neutrality, but these assumptions are not essential.

The sequence of proposers (or \emph{recognition order}) $P = (p_1,...,p_T)$ is assigned at random according to some probability distribution $\mu$. The parties learn about the recognition order as negotiations proceed. We consider three informational protocols.
For the \emph{Perfect} protocol, the parties learn proposers' identities one period in advance; in other words, $p_{t}$ is disclosed between rounds $t-2$ and $t-1$, and the parties start out with knowledge of $p_1$ and $p_2$. For the \emph{Partial} protocol, they learn the identity of one randomly selected \emph{non-proposer} one period in advance: the parties start out knowing $p_1$  and that some $j$ will not be the period-$2$ proposer; $j \neq p_{t}$ is then disclosed between rounds $t-2$ and $t-1$. For the \emph{None} protocol, the parties have no advance knowledge of proposers' identities. In other words, they start out knowing only $p_1$, and $p_{t}$ is disclosed between rounds $t-1$ and $t$. 

This framework encompasses a wide range of possibilities. For example, most theoretical and experimental extensions of \citet{baron1989bargaining} study the \emph{None} protocol in settings where the selection of each round's proposer is independent and, typically, the probabilities are the same for all three parties. 
\cite{ali2019predictability} model the effects of predictability in legislative bargaining, considering a range of informational protocols that include the partial and perfect protocols described above.

As is standard in this literature, we study pure-strategy subgame-perfect equilibria with the following features. First, in each round, strategies prescribe choices that do not depend on the features of previously rejected proposals or the associated profile of votes. Second, when casting votes, non-pivotal voters resolve indifference as if they are pivotal.\footnote{Equivalently, they do not make weakly dominated choices at the voting stage.}  Conditional on the recognition order $P$, any such equilibrium induces a vector of expected continuation payoffs, $(v^A_t, v^B_t, v^C_t)$, that the parties receive if they reject the round-$t$ proposal. If $v^i_t<v^j_t$ for $i,j \neq p_t$, we say that $i$ is the \emph{weak round-$t$ partner}, and that $j$ is the \emph{strong round-$t$ partner}. 

Some portions of our empirical analysis focus on the match between observed behavior and specific quantitative predictions for individual bargaining games, while other portions test a small collection of less precise predictions that are arguably easier for negotiators to grasp. Next, we describe three theoretical principles that encompass those implications. 

\bigskip

\noindent\textit{\textbf{Theoretical Result 1}: In equilibrium, proposals always assign the proposer a share of at least 0.5 and at least one other party a share of 0. Strong partners never receive positive shares. All equilibrium proposals pass.}

\bigskip

The logic of Theoretical Result 1 is straightforward and extremely general. A round-$t$ proposal passes as long as $x^j_t \geq v_t^j$ for at least two parties. The best successful proposal is plainly one that secures passage while conceding no more than necessary. The proposer accomplishes this objective by offering $x^i_t = v^i_t$ for $i = \arg \min_{i \ne p_t} v^i_t$, $x^j_t = 0$ for $j \notin \{p_t,i\}$, and $x^{p_t}_t = 1 - v^i_t$ for themselves. Because $v_t^i \leq v_t^j$, we know that $v_t^i \leq 0.5$, so $x_t^{p_t} \ge 0.5$. Furthermore, because $1-v^i_t \geq v^{p_t}_t$, the best successful proposal delivers a higher share for the proposer than any unsuccessful proposal.    

A common interpretation of this first result is that proposers form minimum winning coalitions (MWCs) with weak partners, to the extent they exist. If $v_t^i = 0$, the proposer offers nothing to the other parties, which somewhat obscures the MWC interpretation. In the associated equilibrium, party $i$ resolves their indifference in favor of the proposal, thereby joining the MWC. Were $i$ to resolve her indifference differently, the proposer would offer $i$ a small token share, thereby securing almost all of the prize.

\bigskip

\textit{\textbf{Theoretical Result 2}: Suppose $p_{t+1} \neq p_t$. When the identity of the next proposer is known one round in advance (as is the case under the Perfect protocol), $p_{t+1}$ is never the weak round-$t$ partner for the proposer $p_t$.}

\bigskip

Theoretical Result 2 is also general and intuitive. From Theoretical Result 1, we know that $p_{t+1}$ will obtain a payoff of at least 0.5 if round $t$ is reached. Because the identity of $p_{t+1}$ is known throughout round $t$, it follows that $v^p_{t+1} \geq 0.5$. But $v_t^i + v_t^{p_{t+1}} \le 1$ for $i \notin \{p_t,p_{t+1}\}$, so $v_t^i \leq v_t^{p_{t+1}}$, which means $p_{t+1}$ is not a weak partner in round $t$.

An important and subtle implication of this result is that the identity of the weak partner can be completely unrelated to the default payoffs $(v^A_T,v^B_T,v^C_T)$, and solely determined by the recognition order. We conjecture that, in practice, failures to properly identify weak partners may be related to this distinction.

\bigskip

\textit{\textbf{Theoretical Result 3}: Predictability of the recognition order strengthens the proposer's bargaining position by weakening at least one of her potential partners. Perfect predictability of the next proposer enables her to secure the entire prize.}\footnote{For the illustration in the introduction, we made the stronger assumption that the \emph{entire} recognition order is known in advance.}

\bigskip

\citet*{ali2019predictability} make the second point with considerable generality. Here we illustrate it for the special case in which (i) the selection of each round's proposer is independent and based on a distribution that assigns equal probability to all three parties, and (ii) all players' default payoffs are zero.
We focus on symmetric equilibria in which, whenever a proposer is indifferent between two potential coalition partners, she chooses each with equal probability.\footnote{In an infinite-round setting, \citet*{baron1989bargaining} show that this property follows from stationarity.}

First imagine that the recognition order is entirely unpredictable (the \emph{None} protocol). In that case, the problem's symmetry and our focus on symmetric equilibria ensures that $(v_t^A,v_t^B,v_t^C)=(\frac{1}{3}, \frac{1}{3}, \frac{1}{3})$ for all $t < T$. It follows that, for such {t}, the round-{t} proposal assigns $\frac{2}{3}$ to the proposer and $\frac{1}{3}$ to a randomly selected partner.

Next imagine that the identity of each proposer is known one round in advance (the \emph{Perfect} protocol). With default payoffs of $(0,0,0)$, the proposer in the final round claims the entire prize. The proposer in the penultimate round can therefore always identify a potential partner who will receive zero in the continuation---specifically, a party who will not make the last offer. Consequently, the penultimate proposer also claims the entire prize. Through a repeated application of the same argument, we see that the proposer claims the entire prize in every round with the support of a party who does not propose next. Comparing this outcome to the one that emerges under the \emph{None} protocol, we see that predictability of the recognition order augments proposer power, effectively turning each proposer into a dictator.

According to \citet*{ali2019predictability} (Theorem 1), the preceding conclusion holds even when information about the next proposer is imperfect, as long as it always rules out more than half of the parties. An ability to rule out only one party (the \emph{Partial} protocol) does not satisfy this requirement. Consequently, theory predicts that the proposer should fare substantially better with the \emph{Perfect} protocol than with the \emph{Partial} protocol.

To refine this prediction, we solve for the equilibrium of a three-round bargaining game under the \emph{Partial} protocol. The round-3 proposer always obtains the entire prize. There are two possibilities for the round-2 proposer. If she learns she will not be the final round-3 proposer (a case we henceforth label \emph{Excl}), then her two potential coalition partners each have an expected continuation payoff of $\frac{1}{2}$. Therefore, her optimal strategy is to offer $\frac{1}{2}$ to a randomly selected partner, leaving a share of $\frac{1}{2}$ for herself, in which case each non-proposer receives an expected payoff of $\frac{1}{2} \times \frac{1}{2}=\frac{1}{4}$. If she learns she might be the next proposer (a case we henceforth label \emph{Incl}), then her best option is to claim the entire prize by forming a coalition with the non-proposer who will not be the round-3 proposer. Consequently, the expected payoff for the round-2 proposer is $\frac{1}{3} \times \frac{1}{2} + \frac{2}{3} = \frac{5}{6}$, and the expected payoff for each round-2 non-proposer is $\frac{1}{12}$. Likewise, there are two possibilities for the round-1 proposer. If she learns she will not be the next proposer (the \emph{Excl} case), then her potential partners each have an expected continuation share of $\frac{1}{2} \times \frac{5}{6} + \frac{1}{2} \times \frac{1}{12} = \frac{11}{24}$. Her best option is to claim a share of $\frac{13}{24}$ and offer $\frac{11}{24}$ to a randomly selected partner. If she learns she might be the next proposer (the \emph{Incl} case), then one of her potential partners has an expected continuation fraction of $\frac{1}{12}$, which means she can claim $\frac{11}{12}$. Overall, her expected payoff is $\frac{1}{3} \times \frac{13}{24} + \frac{2}{3}  \times \frac{11}{12} = \frac{55}{72}$, which substantially exceeds $\frac{2}{3}$. 

The preceding reasoning demonstrates that, even though a small degree of recognition-order predictability does not render the first proposer a dictator, she still obtains a higher expected payoff with the \emph{Partial} protocol than with the \emph{None} protocol. The higher payoff materializes only when she learns she might be the second proposer (the \emph{Incl} case); in fact, she is worse off relative to the \emph{None} protocol when she learns she will not be the second proposer (the \emph{Excl} case), but she gains on average. In contrast, with the \emph{Perfect} protocol, her payoff does not depend on whether she learns she definitely will be the next proposer (which serves as the \emph{Incl} case for the \emph{Perfect} protocol), or definitely will not be the next proposer (which serves as the \emph{Excl} case for the \emph{Perfect} protocol). This difference is attributable to the fact that she can identify a weak partner even in the \emph{Excl} case for the \emph{Perfect} protocol (i.e., the other non-proposer for the next round), but not with the \emph{None} protocol.    

\subsection{Why Complexity May Lead To Egalitarianism}\label{Subsection-egalitarianism}

The strategic reasoning underlying equilibrium predictions becomes increasingly complex as the number of stages and subcases proliferate and may be too subtle for typical negotiators to grasp. Even savvy negotiators who appreciate the logic of equilibrium may lack confidence that others share their strategic sophistication. For instance, they may worry that potential partners are excessively optimistic or pessimistic about their own continuation payoffs.  Consequently, a key issue that we investigate experimentally is whether variations in the complexity of the negotiating environment affect the accuracy of the theoretical predictions, leading systematically to alternative behavioral patterns. In particular, if complexity increases a proposer's uncertainty about her ability to exploit coalition partners and makes her more concerned about her vulnerability to exploitation by others, offering more egalitarian proposals may be a tempting safe alternative. In that case, changes to the environment that add complexity could encourage more egalitarian behavior.

We formalize this point using a simple framework in which a proposer is uncertain about the offers other players would support. Imagine the current-round proposer (Player A) believes that player $i\in \{B,C\}$ has an acceptance threshold $\tau_{i}$ such that $i$ votes for the proposal iff $x_{i}$ is at least $\tau_{i}$.\footnote{Thus, players resolve indifference in favor of supporting the proposal. In addition, each player's acceptance decision depends only on her own share.} One interpretation of $\tau_i$ is that it reflects player $i$'s perceived continuation value. In complex settings, the proposer (player $A$) may be highly uncertain about $\tau_B$ and $\tau_C$ because she recognizes that each player's perceptions of continuation values may stray far from equilibrium predictions. We model this uncertainty by assuming that the proposer's beliefs about others' thresholds correspond to some joint CDF that allows for arbitrary correlations. The proposer also believes that she will receive an expected payoff of $d<1$ in the continuation game if her offer is rejected. Finally, we assume the proposer
is risk neutral.

We are interested in behavior near a boundary case in which the game's
complexity has two effects on the proposer. First, it causes the proposer to hold diffuse beliefs reflecting
deep uncertainty about the other players' thresholds. Second, it leads the proposer to either (i)
narrowly frame the problem by focusing on the stage game in isolation,
or (ii) hold pessimistic beliefs about the continuation game. We
model diffuse beliefs as involving a bivariate CDF $H$ for which the marginal
distributions of each $\tau_i$ ($i=B,C)$, written $H_{i}$, are uniform
on $[0,\overline{\tau}]$, where $\overline{\tau}\geq\frac{1}{2}$. The assumption
concerning the magnitude of $\overline{\tau}$ allows for the possibility that a coalition partner might insist on equal division within an MWC but does not necessarily require a belief that the partner might insist on a share close to unity. In addition to having the intuitive interpretation that all possibilities are equally likely, the uniform distribution reflects ``maximal uncertainty'' because it achieves the highest entropy among all continuous distributions on its support. Consequently, other studies have also used it to represent minimally informed beliefs.\footnote{For example, in the literature on levels of strategic thinking, the beliefs held by higher level players about a level-$0$ player's strategy are often assumed to be uniform. See, for example, \cite{camerer2004cognitive} and \cite{crawford2013survey}.} Notably, this boundary case allows for correlation, positive or negative, between the two thresholds. The following assumption places a modest restriction on this correlation:
\begin{assumption}\label{Assumption-Thresholds}
    For every $x>0$, $H(x,x)>0$. 
\end{assumption}
\Cref{Assumption-Thresholds} states that for every strictly positive value $x$, there is a strictly positive probability that players $B$ and $C$ would both support the offer $(1-2x,x,x)$. Notably, this assumption rules out perfect negative correlation between players' thresholds. It is always satisfied when $H$ has full support on $[0,\overline\tau]\times[0,\overline\tau]$. 

We examine sequences $(H_n,d_n) \rightarrow (H,d)$ where, for the purpose of defining convergence, we endow the space of CDFs with the weak topology. In stating our proposition, we follow the convention that player $B$ is the individual to whom
the proposer offers the weakly larger share ($x_{B}\geq x_{C})$. The following proposition shows that optimal offers in the neighborhood of $(H,d)$ involve minimum winning coalitions.

\begin{proposition}\label{Proposition-Egalitarian}
As $(H_n,d_n) \rightarrow (H,d)$, the optimal offers converge to $\left(\dfrac{1+d}{2},\dfrac{1-d}{2},0\right)$.
\end{proposition}

We are especially interested in the possibility that, in response to the problem's complexity, the proposer either narrowly frames the problem by focusing exclusively on the stage game, or makes highly pessimistic assumptions about the continuation game. Such possibilities correspond to the boundary case of $d=0$. Proposition \ref{Proposition-Egalitarian} implies that as we approach $(H,0)$, 
the optimal offer converges to complete egalitarianism within an MWC, i.e., $\left(1/2,1/2,0\right)$. Even if the proposer is not completely pessimistic about her continuation
payoff $d$, her optimal offer may still involve a reasonably equal split.
For instance, if she expects a share of 20\% in
the continuation game, and she has diffuse beliefs about others'
acceptance thresholds, she will propose a 60-40 split in an MWC.

While \Cref{Proposition-Egalitarian} is intuitive, its proof requires attention to various technical details.\footnote{We benefited from discussing this argument with Andreas Kleiner.} We outline the key steps here, relegating the complete argument to \Cref{Appendix-Theory}. One key step is to show that, for the limit case $(H,d)$, the optimal offer takes the form described in the proposition.\footnote{For this step, the restriction to three-player bargaining games plays a role. With more than three players, one can show that the optimal MWC offer is egalitarian when $d=0$ if players' thresholds are distributed independently. However, offers that seek to form larger coalitions may do better, at least for large $\overline{\tau}$. For example, with five players, $\overline{\tau}=1$, and $d=0$, an offer with equal shares among three players (including the proposer) yields an expected payoff of 0.0370, while offering shares of 0.22 to three non-proposers yields an expected payoff of 0.0420.} In another key step, we show that, for arbitrary beliefs, $(H,d)$, the proposer's payoff function is upper semicontinuous in the offer, which guarantees existence of an optimal offer. Further, we show that, in a neighborhood of $(H,d)$, the optimal offer correspondence is upper hemicontinuous in the proposer's beliefs about other players' thresholds and her own expected payoff in the continuation game following a rejection.\footnote{We establish this result using a variant of the Theorem of the Maximum that holds for upper semicontinuous payoff functions. Berge's Theorem  does not apply because the proposer's payoff function is discontinuous.} Therefore, when beliefs are close to the boundary case of $d=0$, egalitarian offers targeting minimum winning coalitions are approximately optimal.

The preceding analysis shows how deep uncertainty about others' thresholds coupled with narrow framing or pessimism about the continuation game can give rise to egalitarian MWC offers. Given that we also observe grand coalition (GC) offers in our experiment, one may wonder why proposers also find those offers attractive. We highlight three potentially relevant considerations. 

First, strategic uncertainty can substantially reduce the difference in expected payoffs from MWC and GC offers. For example, while the difference between the proposer's payoffs in the egalitarian MWC and GC allocations is $0.17$ ($0.5$ versus $0.33$), the difference in expected payoffs from the corresponding offers with uniform-independent priors, $\overline{\tau}=1$, and $d=0$ is only $0.065$ ($0.25$ versus $0.185$), partly because the GC offer provides a hedge. A preference for fairness is, therefore, more likely to outweigh the expected monetary gain, especially if that preference pertains to the offer itself rather than to the implemented allocation---for example, because acting fairly has deontological force as in \cite{Andreoni2020} or signaling value as in \cite{Andreoni2009}. 

Second, proposers may believe that the fairness of an outcome elevates the probability of acceptance. In other words, there may be a consensus that equal division is reasonable in settings that feel mostly random due to their complexity. Notably, this consideration can also reinforce the attraction of egalitarianism within MWCs.

Third, despite understanding that one potential partner is likely weaker than the other, a proposer may not be sure which is which. Thus, she may believe that $\tau_i$ and $\tau_j$ are negatively correlated. This possibility increases the value of the hedge embedded in a GC offer. To illustrate, suppose a proposer with $d=0$ believes $\tau_i$ is uniformly distributed on $[0,1]$, and that $\tau_j=1-\tau_i$. In that case, when offering $s_i$ and $s_j$, the proposer's expected payoff is $(1-s_i-s_j)(s_i+s_j)$. Because the two first-order conditions are redundant, the optima consist of the locus $s_i=\frac{1}{2} - s_j$. Consequently, offering $\frac{1}{2}$ to one party and nothing to the other is still optimal, but so is offering $\frac{1}{4}$ to each party. Notably, the latter option minimizes the Gini coefficient among all proposals yielding the maximized payoff. Consequently, if the proposer has a tiny preference for fairness or considers such  offers slightly more likely to be accepted, a Grand Coalition offer emerges as the most attractive alternative. Furthermore, the associated difference between the expected payoff from an optimal proposal and an egalitarian GC proposal falls to $0.028$ ($0.250$ versus $0.222$), making the latter option even more attractive for those who care about fairness even to a small degree.

Taken as a whole, the analysis of this section underscores an important conceptual point: with strategic uncertainty, egalitarianism among MWCs, GCs, or both may become increasingly attractive in complex settings.

\section{Experimental Design}\label{Section-Design}

The experiment examines behavior in three-person finite-round bargaining games under majority rule.\footnote{We also ran some small-scale pilots involving indefinite-horizon games with discounting and communication among committee members. The data from these treatments are available upon request.} In each game, participants negotiate the division of a fixed prize, denominated in "points," which are convertible to dollars at a known rate.\footnote{For example, in some treatments the participants divide a prize of 240 points. In that case, a share of $0.90$ corresponds to $216$ points.} The games we consider differ with respect to the number of rounds and the informational protocol.\footnote{There are some other inessential differences in details across the treatments, including the total number of points and the conversion rate from points to dollars.}

For three-round negotiations, we examine all informational protocols mentioned in the preceding section (\emph{Perfect}, \emph{Partial}, and \emph{None}). As in our examples, we set all the default payments equal to zero. 

For the two-round negotiations, all parties know the identities of both proposers at the outset. We set default shares to 0.05 for the first-round and second-round proposers (who always differ), and to 0.90 for the non-proposer. Our objective is to determine how the players would behave in a three-round bargaining problem---one with default payments of zero---if they thought through the third-round outcome by performing one round of backward induction (approximately) correctly. By comparing behavior between these two-round games and the \emph{Perfect} three-round games, we can determine the extent to which departures from theoretical predictions in the latter result from difficulties associated with thinking forward to third-round outcomes.

For the one-round negotiations, there is no distinction between the three informational protocols. We set default shares to 0.05 for the first-round proposer and one of the non-proposers, and to 0.90 for  the other non-proposer. Our objective is to determine how the players would behave in a three-round \emph{Perfect} bargaining problem if they thought through the second-round and third-round outcomes by performing two rounds of backward induction (approximately) correctly. By comparing behavior between the one-round and three-round games, we can determine the extent to which departures from theoretical predictions in the latter result from difficulties associated with thinking forward to both second-round and third-round outcomes.

Notably, in the one-round bargaining problems, it is reasonably obvious that the weak potential partner is the one with the smaller default share (0.05). In contrast, in the two-round bargaining problems, the weak potential partner is the one with the \emph{larger} default share (0.90). The first proposer must reason out that strength and weakness are consequences of the recognition order (specifically, who proposes second) rather than the default payments.

 \Cref{table:AllPredictions} summarizes the equilibrium offers for each experimental treatment. We explained the predictions for three-round bargaining games in the previous section. The logic behind the implications listed for the one-round and two-round bargaining games is straightforward. We attach the \emph{Perfect} label to the one-round and two-round games to underscore their conceptual connection to the three-round \emph{Perfect} games.

\begin{table}[h!]
\caption{Equilibrium Offers in Bargaining Games}
\label{table:AllPredictions}
\begin{center}\vspace{-.15in}
\begin{tabular}{l|cccc}
\hline\hline
Game & Round 1 & Round 2 & Round 3 \\ \hline
1-round, \emph{Perfect} & 0.95$^*$ & NA & NA \\ \hline
2-round, \emph{Perfect} & 1.0$^*$ & 0.95$^*$ & NA \\ \hline
3-round, \emph{Perfect}   & 1.0$^*$ (\emph{Incl}) & 1.0$^*$ (\emph{Incl}) & 1.0 \\
 & 1.0$^*$ (\emph{Excl}) & 1.0$^*$ (\emph{Excl}) &  \\ \hline
3-round, \emph{Partial} & 0.92$^*$ (\emph{Incl}) & 1.0$^*$ (\emph{Incl}) & 1.0 \\
 & 0.54 (\emph{Excl}) & 0.50 (\emph{Excl}) &  \\ \hline
3-round, \emph{None}   & 0.67 & 0.67 &  1.0 \\ \hline\hline
\end{tabular}
\end{center}
\par
\vspace{2mm} {\footnotesize \underline{Notes:} The table shows the share of the prize the proposer seeks to keep for herself. The proposer offers the residual share to a single partner. $^*$ indicates that the proposer targets a weak partner.}
\end{table}

We conducted twelve sessions at UCI and nine at Caltech, each drawing on the standard student population of subjects. Each session consisted of either 15 or 20 repetitions of the same bargaining game. We refer to each bargaining game as a \emph{treatment}, and to each repetition as a \emph{match}.  \Cref{table:Design} summarizes the numbers of matches, sessions, and subjects for each treatment at UCI and Caltech. 

\begin{table}[h!]
\caption{Treatments and Sessions}
\label{table:Design}
\begin{center}\vspace{-.15in}
\begin{tabular}{l|l|ccc}
\hline\hline
 & Label & \# Sessions & \# Matches & \# Subjects \\ \hline
\textbf{UCI}   &  &  &  &  \\
  1-round, \emph{Perfect} & 1-Perfect   & 3 & 15 & 57 \\
  2-round, \emph{Perfect} & 2-Perfect   & 3 & 15 & 57 \\
  3-round, \emph{Perfect} & 3-Perfect  & 2 & 20 & 42 \\
  3-round, \emph{Partial} & 3-Partial  & 2 & 20 & 48 \\
  3-round, \emph{None}    & 3-None    & 2 & 20 & 48 \\
 \hline
\textbf{CALTECH}   &  &  &  &   \\
  1-round, \emph{Perfect} & 1-Perfect   & 3 & 15 & 33 \\
  2-round, \emph{Perfect} & 2-Perfect  & 3 & 15 & 33 \\
  3-round, \emph{Perfect} & 3-Perfect  & 3 & 20 & 39 \\ \hline
\end{tabular}
\end{center}
\par
\vspace{2mm} {\footnotesize \underline{Notes:} The 1-round and 2-round treatments were conducted within the same session at both
locations (Part I and Part II of the same experiment, respectively). }
\end{table}

The number of subjects for each session was a multiple of 3, and ranged from 9 to 24. For the first match of each session, we randomly divided the subjects into $3$-member groups, which we call \emph{committees}. In each subsequent match, we reshuffled the groupings randomly so that the composition of the committees changed unpredictably over the course of a session.

At the beginning of each match, we randomly assigned each of the three committee members an ID number (1, 2, or 3), which they kept for the duration of the match. We randomly assigned new ID numbers in subsequent matches. Committee member 1 always made the first proposal. For the 2-round treatment, we told subjects at the outset of the match that committee member 2 would be the second proposer. For the 3-round \emph{Perfect} treatments, we publicly announced the ID number of the next-round proposer if the current proposal failed. For the 3-round \emph{Partial} treatments, we publicly announced the ID numbers of two players who were eligible to be the next-round proposer if the current proposal failed. For the 3-round \emph{None} treatments, we provided no information about the identity of the next-round proposer.

\section{Results}\label{Section-Results}

In the analyses that follow, we focus on the second half of all treatments, i.e., the last eight matches in the 1-Perfect and 2-Perfect treatments and the last 10 matches in the 3-None, 3-Partial, and 3-Perfect treatments. We impose this filter because our interest is in behavior after subjects have had the opportunity to experience the game and converge to a stable behavior. We refer to these matches as \textit{experienced}.\footnote{Analyses for all matches, both the first and second halves of each treatment, appear in \Cref{Appendix-AllMatches}. The behavior is largely similar.} 
Statistical tests are performed using regression analyses, in which we cluster standard errors at the session level to account for the interdependencies that come from re-matching subjects between bargaining games. To compare average outcomes between treatments, we regress the variable of interest on a constant and an indicator for one of the two treatments. The p-values we report are those associated with the estimated coefficient on the dummy for the treatment variable.

\subsection{Bargaining Outcomes}\label{sec:Outcomes}

We begin by describing, at a coarse level, the observed bargaining outcomes. We first examine the extent to which the first proposer successfully appropriates surplus (\Cref{sec:OutcomesFirstProposer}) and then turn to the distribution of payoffs (\Cref{sec:OutcomesDistr}). 

\subsubsection{How Successful Is The First Proposer?}\label{sec:OutcomesFirstProposer}

 \Cref{fig:FirstProposer} depicts the realized share of resources appropriated by the first proposer in each treatment, focusing on first-round proposals that were accepted.\footnote{We limit the sample to matches in which the first offer is accepted because including payoffs for matches in which the first offer is rejected confounds the continuation values in multi-round games with the exogenously specified default values for one-round games. That said, the qualitative patterns noted below are similar when we include final payoffs for all matches; see \Cref{Appendix:AlternateFigs}.} For comparability, we organize the treatments by location. number of rounds, and bargaining protocol.

\begin{figure}[h!]
    \begin{center}    
    \caption{Realized Proposer Share of Accepted First-round Offers}
    \label{fig:FirstProposer}
    \includegraphics[scale=0.6]{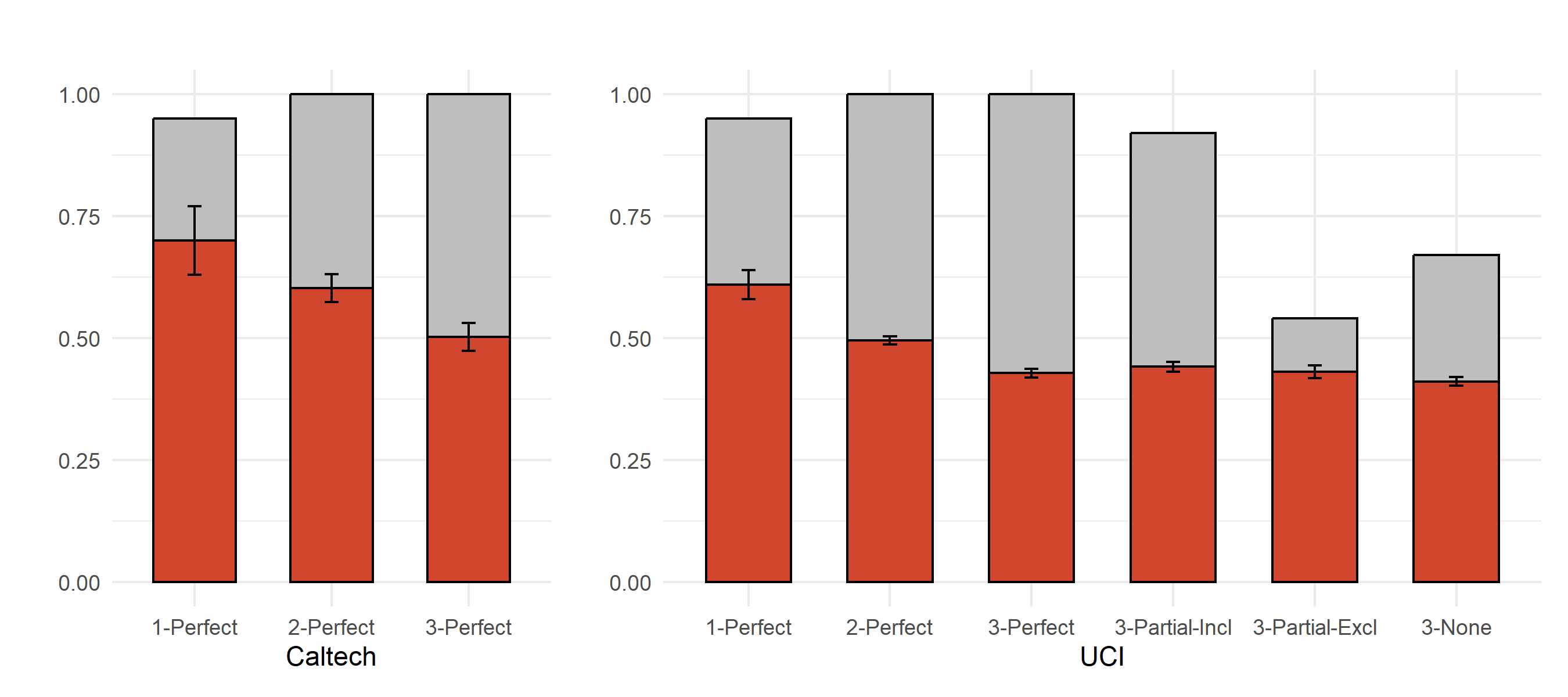}
    \end{center}

    {\footnotesize \underline{Notes:} The total height of each bar represents the predicted equilibrium share for the first proposer while the darker portion is the empirical average appropriated by first proposers in the experiment. The error bars mark the 95\% confidence intervals computed based on robust standard errors (clustered at the session level). The empirical averages and standard errors used in this figure are reported in \Cref{table:fig1-realized-share-first-proposer} of \Cref{Appendix-TablesBar}.}
\end{figure}

Three features of this figure are especially salient. First, standard equilibrium theory predicts much higher first-proposer shares than we observe. For example, theory implies that the first proposer ought to obtain all or nearly all of the surplus in the treatments with perfectly predictable recognition orders. However, at both locations, the average share is always well below 75\%, and often much lower. While it is always higher at Caltech than at UCI, the gap between theory and observation is large at both locations. We explore the reasons for the difference between outcomes at UCI and Caltech in subsequent sections. 

Second, standard equilbirium theory fails to anticipate a striking comparative static: for Perfect games, the first proposer's share declines significantly as the potential number of bargaining rounds rises at both Caltech and UCI. It turns out that this pattern is symptomatic of a more general phenomenon involving inequality of the payoff distribution across committee members. We turn to that topic in the next subsection.

Even if proposers do not fully understand that, in equilibrium, their bargaining positions are very strong, they might nevertheless recognize that they should be able to obtain a share of at least 50\% (Theoretical Result 1 in \Cref{Section-Games}). Consistent with this principle, the average first proposer's share is always at least 50\% in all Caltech treatments, although it does fall to that level in 3-Perfect treatment. However, behavior at UCI is inconsistent with this principle: the first proposer's average share falls below 50\% in all four three-round treatments as well as the 2-Perfect treatment. It nevertheless remains above 40\%, so proposers are able to exploit some advantage on average in all our treatments. 

Third, we find no support for the theoretical implication that greater predictability of the recognition order should strengthen the first proposer's bargaining position and enable them to obtain a larger share (Theoretical Result 3). No such pattern is visible across the last four bars of \Cref{fig:FirstProposer}: the first proposer's share hovers between 41\% and 44\%, and the variations do not align with theoretical predictions. 

\subsubsection{Payoff Inequality Across Committee Members}\label{sec:OutcomesDistr}

We measure the inequality of payoff distributions across the three committee members using the Gini coefficient. Formally, for a division $(x_1,x_2,x_3)$ of a dollar with shares written in ascending order, the Gini coefficient is given by the formula:
\begin{align*}
    G(x_1,x_2,x_3)\equiv\frac{|x_2-x_1|+|x_3-x_1|+|x_3-x_2|}{3(x_1+x_2+x_3)}=\frac{2}{3}(x_3-x_1).
\end{align*}
Higher values indicate greater inequality. Specifically, the Gini coefficient is $0$ for the completely egalitarian prize distribution, $(\frac{1}{3},\frac{1}{3},\frac{1}{3})$. Equality within a minimal winning coalition, $(0,\frac{1}{2},\frac{1}{2})$, results in a Gini coefficient of $\frac{1}{3}$. The dictatorial outcome $(0,0,1)$ yields a Gini coefficient of $\frac{2}{3}$, the theoretical maximum for this environment. 

\begin{figure}[h!]
    \begin{center}    
    \caption{CDFs of Gini Coefficients of Final Payoffs of First-round Accepted Offers}
    \label{fig:Distr}
    \includegraphics[scale=0.7]{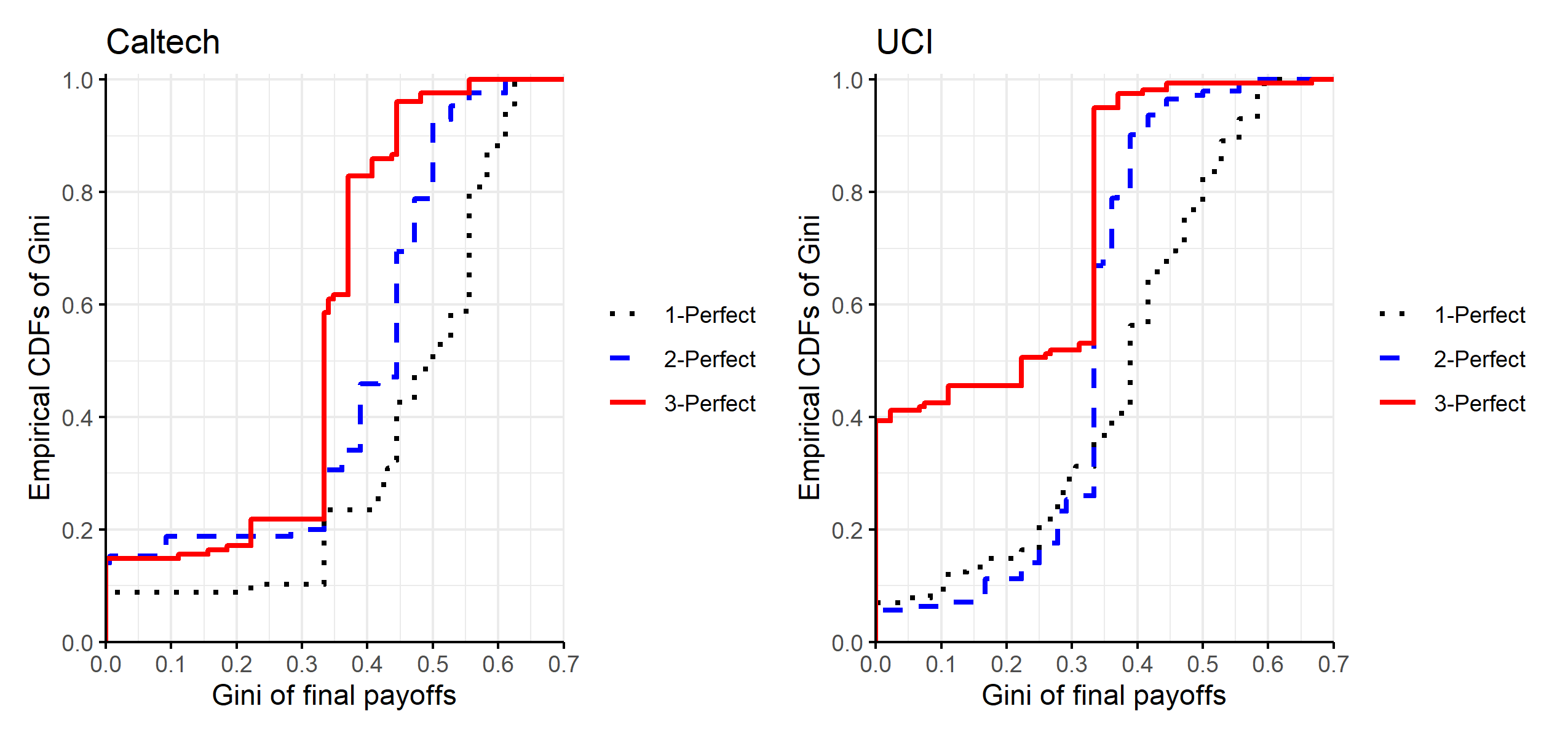}
    \end{center}

    {\footnotesize \underline{Notes:} We present the cumulative distribution functions of Gini coefficients computed based on final payoffs of accepted first-round offers in treatments 1-Perfect, 2-Perfect, and 3-Perfect in both locations. \Cref{fig:fig2-alt-all-offers} in \Cref{Appendix:AlternateFigs} displays qualitatively similar inequality distributions using all (experienced) matches, including matches where the first-round offer failed to pass.}
\end{figure}
    
\Cref{fig:Distr}, which displays the cumulative distribution functions (CDFs) of Gini coefficients separately for 1-Perfect, 2-Perfect, and 3-Perfect bargaining games at Caltech and UCI, allows us to compare the payoff inequality among committee members across different treatments.\footnote{For the same reason as in the previous section, we restrict attention to accepted first-round offers.} The horizon of the game significantly impacts payoff inequality: longer-horizon games generally yield a more equitable distribution of resources among committee members, as indicated by lower GINI coefficients.\footnote{The Kolmogorov-Smirnov statistics for equality of distributions are, for Caltech, 0.4 in 1-Perfect versus 2-Perfect, 0.62 in 1-Perfect versus 3-Perfect, and 0.49 in 2-Perfect versus 3-Perfect; for UCI, 0.39 in 1-Perfect versus 2-Perfect, 0.58 in 1-Perfect versus 3-Perfect, and 0.39 in 2-Perfect versus 3-Perfect; all of these differences are highly significant with p-values below 0.1\%.} 

Strikingly, the CDFs above exhibit some clear mass points (vertical jumps), and the concentrations at these points are plainly related to the length of the game. Looking first at Caltech, equal division within MWCs (i.e., a Gini coefficient of $\frac{1}{3}$) accounts for between from 11\% to 13\% of allocations in one and two-round negotiations versus 37\% of allocations in three-round negotiations. Qualitatively, we observe a similar (albeit less dramatic) pattern for equal division within grand coalitions (i.e., a Gini coefficient of $0$), which accounts for 9\% of allocations in one-round problems, 14\% in two-round problems, and 15\% in three-round problems. At UCI, equal division within MWCs (i.e., a Gini coefficient of $\frac{1}{3}$) accounts for only 5\% of allocations in one-round games versus 41\%-42\% in two-round and three-round games. Furthermore, equal division within grand coalitions (i.e., a Gini coefficient of $0$) accounts for only 6\% to 7\% in one and two round games versus 39\% for the three-round games. Altogether, at Caltech, equal splits within a coalition (either a MWC or a grand coalition) account for 52\% of allocations in three-round games versus 22\% in one-round games; at UCI, the corresponding figures are 81\% and 13\%. Accordingly, moving from one round to three, egalitarianism more than doubles at Caltech and increases by a factor of more than six at UCI. In other words, we see a dramatic migration towards egalitarian splits in games with longer horizon.

The number of rounds in a bargaining game provides a natural proxy for the game's strategic complexity. A longer game requires a subject to think through far more contingencies and to form more elaborate beliefs about how others will reason and behave. Viewed from this perspective, our results offer clear evidence that increased complexity pushes subjects towards more egalitarian allocations. However, these results pertain only to the final outcomes of negotiations. In the next section, we provide more detailed analyses of first-round proposals, voters' responses to these proposals, and the optimality of proposals in light of voters' behavior. 

\subsection{How Do Groups Arrive At These Bargaining Outcomes?}

\subsubsection{First-Round Proposals}\label{sec:Proposals}
In this section, we analyze all first-round proposals, including those that do not pass. This analysis uncovers the extent to which proposers target and exploit weak coalition partners, and how that exploitation depends on the complexity of the game. 

We categorize proposals using the following taxonomy. A \emph{coalition partner} is a non-proposing player to whom the offer assigns at least 5\% of the budget. A \emph{minimal winning coalition} (MWC) proposal is one that includes exactly one coalition partner, a \emph{grand coalition} proposal is one that includes two coalition partners, and a \emph{dictatorial} proposal is one that includes no coalition partners.\footnote{Others have employed similar criteria for categorizing coalitions; see, for example, \cite{frechette2005nominal} and \cite{agranov2014communication}. Our results are not sensitive to reasonable alternatives to the $5\%$ threshold. There is some ambiguity as to how we should classify proposals that offer both non-proposing partners less than 5\%. If either or both non-proposing players have a continuation payoff of zero, then one could construe the proposal that assigns the entire prize to the proposer  as an MWC or grand coalition offer (assuming that either or both non-proposing players resolve indifference favorably). In practice, dictatorial proposals are rare, accounting for fewer than 3\% of proposals in any of our treatments.} We classify an MWC offer as targeting a weak (respectively, strong) coalition partner if, according to theory, the partner has a lower (respectively higher) continuation payoff than the excluded player. If theory implies that both non-proposing players have the same continuation payoff, we classify the coalition as MWC-NA.

\Cref{fig:Coalitions} shows distributions of proposals across these categories for each treatment. The top panel depicts results for Caltech subjects. In one-round games, the vast majority of proposals (86\%) seek to form MWCs, and all of them target weak partners. These findings confirm that Caltech subjects have exploitative intent; they also show that these subjects have exploitative ability in the simplest bargaining contexts. The same fraction of proposals seek to form MWCs in two round games, and 88\% of those target weak partners. In other words, proposers try to behave much as they do in one-round games, but occasionally choose the wrong partner. Behavior changes more significantly in three-round games: 77\% of proposals seek to form MWCs when the first-round proposer is not the second-round proposer (3-Perfect-Excl); when the first-round proposer remains the second-round proposer (3-Perfect-Incl), the figure is 85\%. Equilibrium predictions identify one non-proposing player to be weaker than the other only in the first setting, and 75\% of those proposals target the weak partner. Thus, we see a greater shift toward grand coalition offers in the three-round game alongside a continued decline in exploitative ability.

\begin{figure}[t!]
    \begin{center}
    \caption{Coalition types in first-round proposals}\vspace{-.15in}
    \label{fig:Coalitions} 
    \includegraphics[scale=0.65]{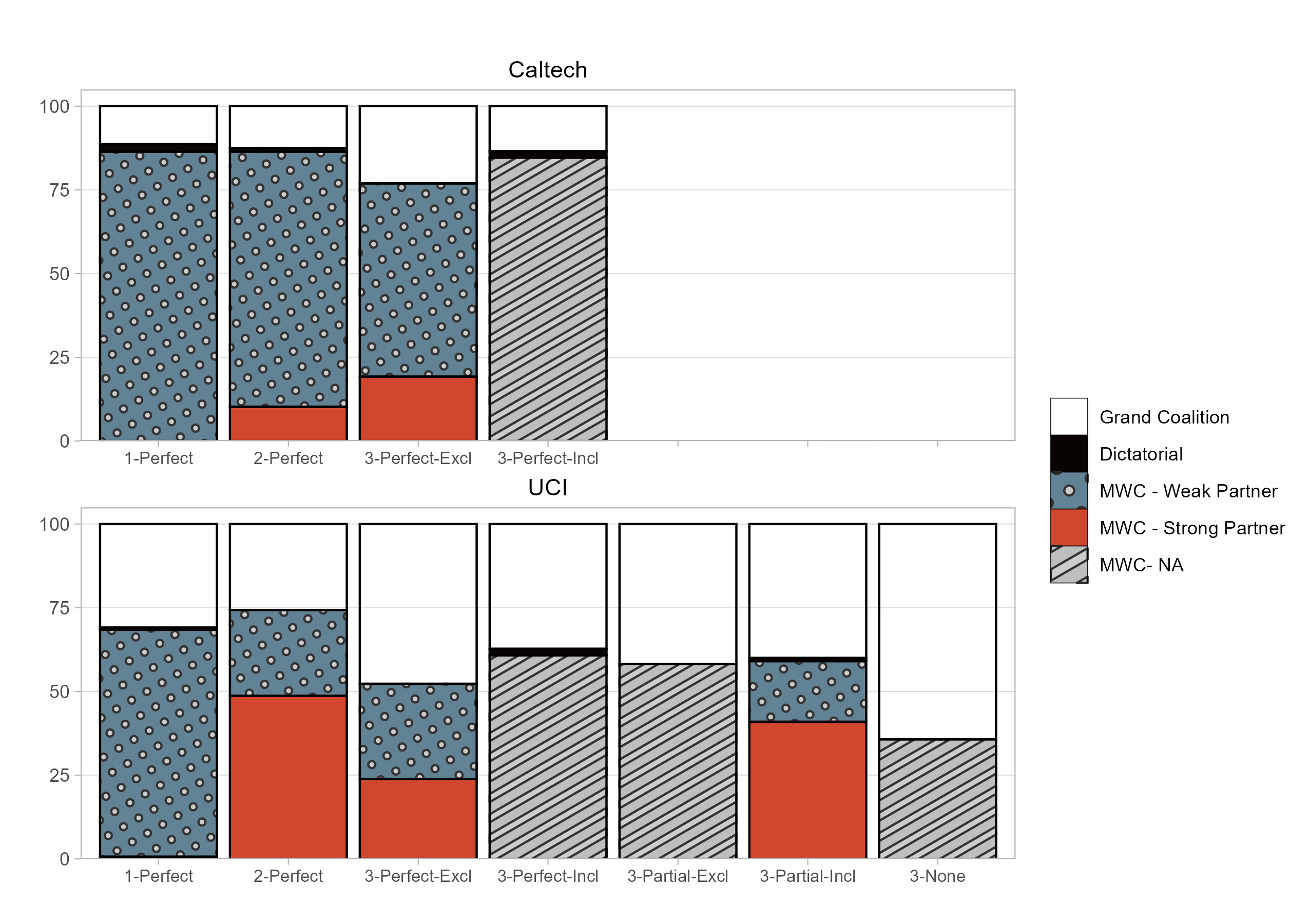} 
    \end{center}

    {\footnotesize \underline{Notes:} The empirical frequencies used in this figure are reported in \Cref{table:fig-3-coaliution-types} of \Cref{Appendix-TablesBar}.}
\end{figure}

We find different patterns at UCI: in relatively simple games, players are much more inclined to make offers that either target grand coalitions, or that mistakenly target strong partners in MWCs. In one-round games, 68\% of the proposals seek to form MWCs, and 99\% of those target weak coalition partners. Consequently, at least half of the UCI subjects have exploitative intent, and they also have exploitative ability in the simplest bargaining contexts. In two-round games, the fraction of MWC offers increases to 74\%, indicating even greater exploitative intent. However, in contrast to Caltech subjects, UCI subjects often fail to realize that the player with the higher default payoff is in the weaker bargaining position because they do not propose next. Consequently, contrary to Theoretical Result 2, almost two-thirds of MWC offers mistakenly target the strong partner. Despite the differences between Caltech and UCI for one-round and two-round games, the patterns of offers exhibit greater similarity for the more complex three-round games. At UCI, 52\% of proposals seek to form MWCs when the first proposer is not the second proposer (3-Perfect-Excl). These proposals target weak and strong partners with roughly equal frequencies, indicating an inability to discern theoretical bargaining strength. When the first proposer is also the second proposer (3-Perfect-Incl), the frequency of MWC offers is slightly higher. The share of proposals that seek to form MWCs is similar for the 3-Partial-Excl and 3-Partial-Incl treatments but falls to 36\% for the 3-None treatment. UCI subjects also tend to select the stronger player, rather than the weaker one, as a coalition partner in 3-Partial-Incl.

\begin{table}[t!]
    \begin{center}
    \caption{Ability to Identify Weak Partners}
    \label{tab:Exploitative}
 {\footnotesize  
\begin{tabular}{@{\extracolsep{5pt}}lcc} 
		\\[-1.8ex]\hline 
		\hline \\[-1.8ex] 
		& P(Strong) - P(Weak) & P(Strong|MWC) \\ 
		\\[-1.8ex] & (1) & (2)\\ 
		\hline \\[-1.8ex] 
		Caltech, 1-Perfect & $-$0.864$^{***}$ & 0.000 \\ 
		& (0.025) & (0.000) \\ 
		& & \\ 
		Caltech, 2-Perfect & $-$0.659$^{***}$ & 0.118 \\ 
		& (0.221) & (0.091) \\ 
		& & \\ 
		Caltech, 3-Perfect-Excl & $-$0.385$^{***}$ & 0.250$^{***}$ \\ 
		& (0.140) & (0.067) \\ 
		& & \\ 
		UCI, 1-Perfect & $-$0.671$^{***}$ & 0.010 \\ 
		& (0.064) & (0.008) \\ 
		& & \\ 
		UCI, 2-Perfect & 0.230 & 0.655$^{***}$ \\ 
		& (0.215) & (0.142) \\ 
		& & \\ 
		UCI, 3-Partial-Incl & 0.229$^{***}$ & 0.694$^{***}$ \\ 
		& (0.034) & (0.027) \\ 
		& & \\ 
		UCI, 3-Perfect-Excl & $-$0.046 & 0.456$^{***}$ \\ 
		& (0.046) & (0.041) \\ 
		& & \\ 
		\hline \\[-1.8ex] 
		\hline 
		\hline \\[-1.8ex] 
		\textit{Note:}  & \multicolumn{2}{r}{$^{*}$p$<$0.1; $^{**}$p$<$0.05; $^{***}$p$<$0.01} \\ 
	\end{tabular} }

\end{center}

    {\footnotesize  \underline{Notes:}  P(strong) (resp. P(weak)) denotes the overall frequency with which first proposers make MWC offers targeting strong (resp. weak) partners. Standard errors are clusterered by session. We omit the treatment 3-Perfect-Incl as both non-proposing players would be weak partners. }
\end{table}

\Cref{tab:Exploitative} summarizes the information presented in \Cref{fig:Coalitions} in a slightly different way, in order to address directly the exploitative ability hypothesis (as well as Theoretical Result 2). For each treatment, column (1) shows the difference between the frequency with which proposers make MWC offers targeting (theoretically) strong players and the frequency with which they make MWC offers targeting (theoretically) weak players. Lower (more negative) numbers indicate greater exploitative ability, zero implies no ability to distinguish between weak and strong players, and numbers greater than 0 reflect a tendency to systematically confuse weak and strong players. At Caltech, subjects exhibit exploitative ability in all treatments, but this ability declines substantially as the environment becomes more complex. At UCI, subjects exhibit exploitative ability only in one-round bargaining games. They commit systematic errors with respect to discerning potential partners’ bargaining strengths in two-round games and some three-round games. Column (2) explores the same issue using a measure of discernment that does not depend on the prevalence of MWC offers: the probability of choosing strong partner conditional on making an MWC offer. The pattern is similar. 

Next, we ask how the tendency to make egalitarian offers evolves as the bargaining environment becomes more complex. Recall that we classify an offer as \emph{egalitarian} if, within a coalition, no two players' shares differ by no more than $5\%$ of the prize. For each treatment, \Cref{fig:egalitarian} shows the distribution of offers across three categories: MWC egalitarian, grand coalition egalitarian, and non-egalitarian. The patterns are striking. At Caltech, the total frequency of egalitarian offers rises from $15\%$ in the one-round treatment to $25\%$ in the two-round treatment, to $60\%$ in the three-round treatment when the first and second proposers differ (3-Perfect-Excl).\footnote{The frequency of egalitarianism is lower ($32\%$) in three-round settings when the first proposer enjoys the advantage of proposing twice in a row (3-Perfect-Incl).} The same qualitative pattern emerges at UCI but the magnitudes are much larger: the total frequency of egalitarian offers rises from $14\%$ in the one-round treatment to $51\%$ in the two-round treatment to $83\%$ in the 3-Perfect-Excl treatment.\footnote{Again, the egalitarian share is a bit lower ($71\%$) when the first-proposer enjoys the advantage of proposing twice in a row (3-Perfect-Incl). The egalitarian share is generally $60\%$ or above for all three-round games in UCI.} Thus, egalitarianism tends to emerge progressively as complexity increases.

\begin{figure}[t!]
    \begin{center}
    \caption{The Frequency of Egalitarian Offers}\vspace{-.15in}
    \label{fig:egalitarian}
    \includegraphics[width=5in]{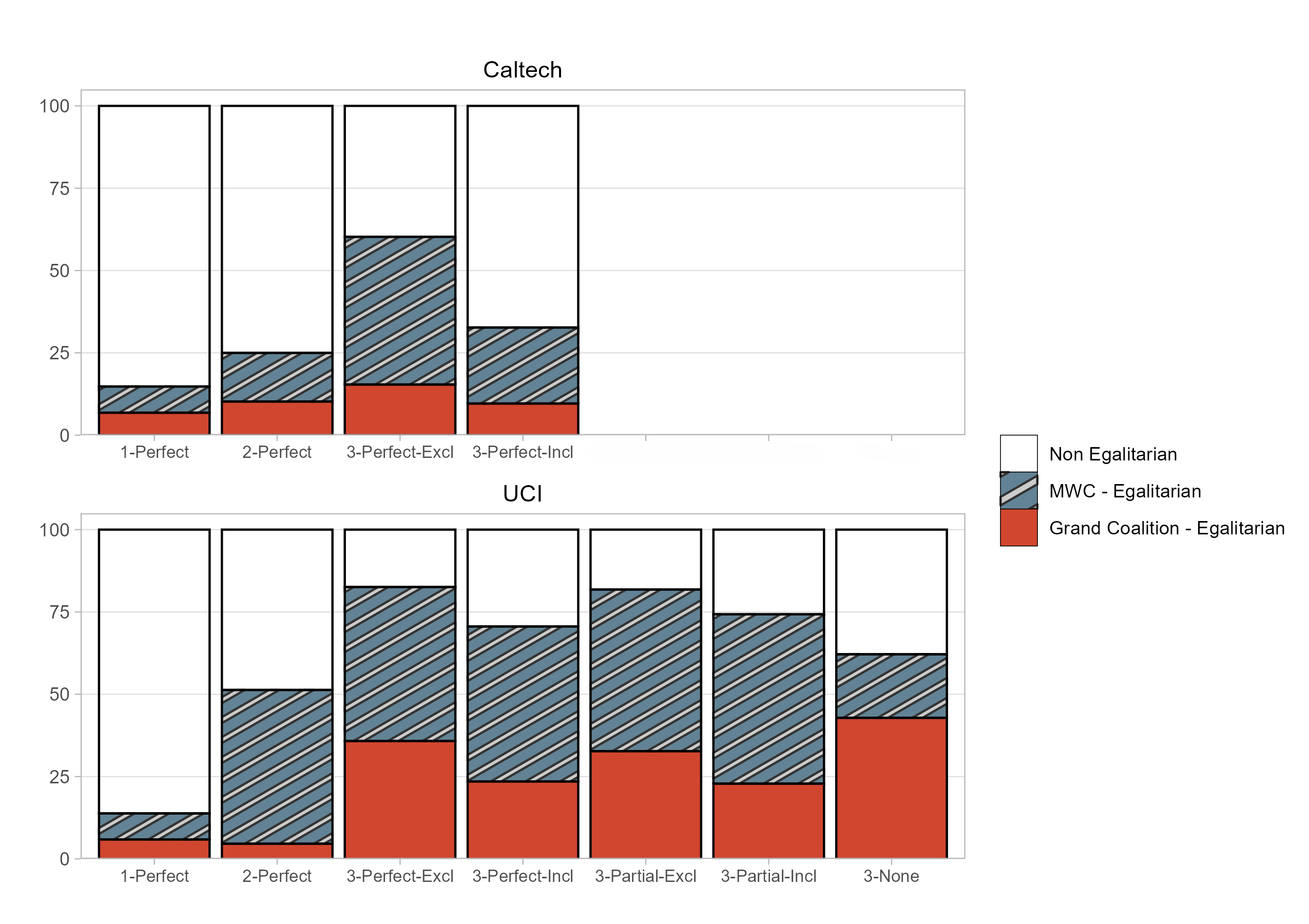}
    \end{center}

    {\footnotesize \underline{Notes:} The empirical frequencies used in this figure are reported in \Cref{table:fig-4-egal-coalition} of \Cref{Appendix-TablesBar}.}
\end{figure}

Notably, we observe the same migration to egalitarianism even among subjects who conform most closely to theory, in the sense that they make MWC offers to weak players. For each pertinent treatment (i.e., those for which non-proposing players are asymmetrically positioned), \Cref{fig:MWCegalitarian} displays the frequency of egalitarianism among MWC offers to weak players, as well as the average share the first proposers hope to keep for themselves. At Caltech, fewer than 10\% of these offers are egalitarian in either one-round or two-round games. Even so, adding the second round leads the first proposers to behave more conservatively, in the sense that the proposer's share declines. With the addition of a third round, the egalitarian-offer share jumps up to $62\%$ and the average proposed share for the first proposer is barely more than $50\%$. Once again the pattern is qualitatively similar at UCI but the quantitative changes are much larger. While 12\% of offers are egalitarian in one-round games, the share jumps to roughly one-third in two-round games and to roughly $90\%$ (resp. $84\%$) in three-round games when the first and second proposers differ (resp. are the same). The proposed share for the first proposer declines from $70\%$ in the one-round game, to $56\%$ in the two-round game, to barely more than $50\%$ in the three-round games. Accordingly, even for the most theory-compliant participants, we see a striking migration toward egalitarianism as the game becomes more complex.

\begin{figure}[t!]
    \begin{center}
    \caption{Frequency of Egalitarianism Among MWC Offers to Weak Partners}
    \label{fig:MWCegalitarian}
    \includegraphics[width=5in]{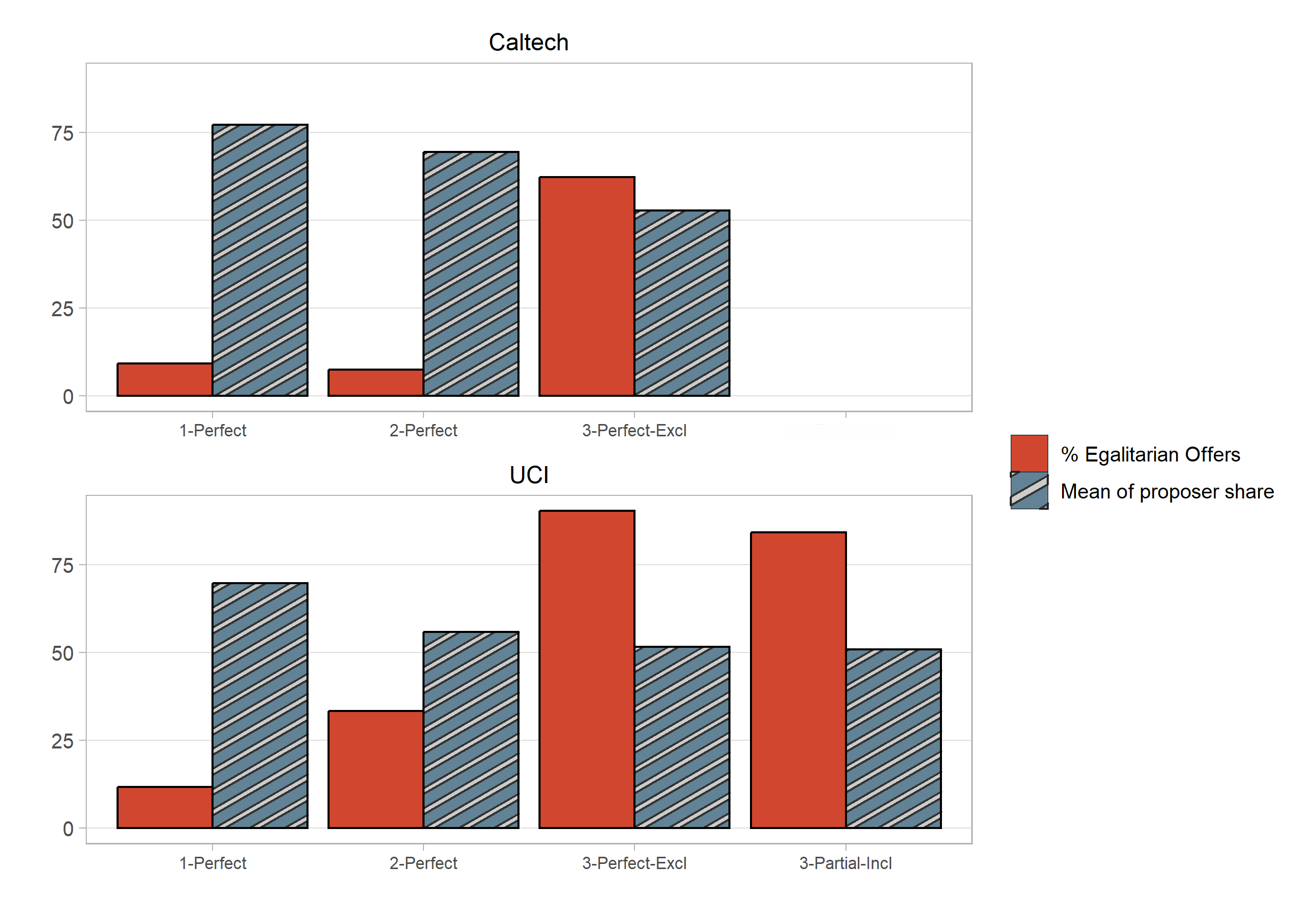}
    \end{center}

    {\footnotesize \underline{Notes:} The empirical frequencies used in this figure are reported in \Cref{table:fig-5-egal-share-offers} of \Cref{Appendix-TablesBar}.}
\end{figure}

\subsubsection{Votes And Coalitions}\label{sec:Responses}

Having completed our analysis of the first proposers' offers, we next examine how others vote on these offers.  \Cref{fig:acceptedrejected} displays the raw data for the one-, two-, and three-round \emph{Perfect} treatments.
In each case, we plot first-round offers on a simplex. The numbers along the right face label the level grid lines for the first proposer's share, which reaches a maximum of $100\%$ at the top vertex. The numbers along the bottom face label the level grid lines for the strong partner's share, which reaches a maximum of $100\%$ at the bottom-right vertex. The numbers along the left face label the level grid lines for the weak proposer's share, which reaches a maximum of $100\%$ at the bottom-left vertex. In each case, a player's share is constant along any of the grid lines emanating from the face that shows their share. The figure distinguishes between accepted proposals (shown as black circles) and rejected proposals (gray circles). The size of each circle reflects the number of observations at that point.

\begin{figure}[p]
    \begin{center}
    \includegraphics[width=5.5in]{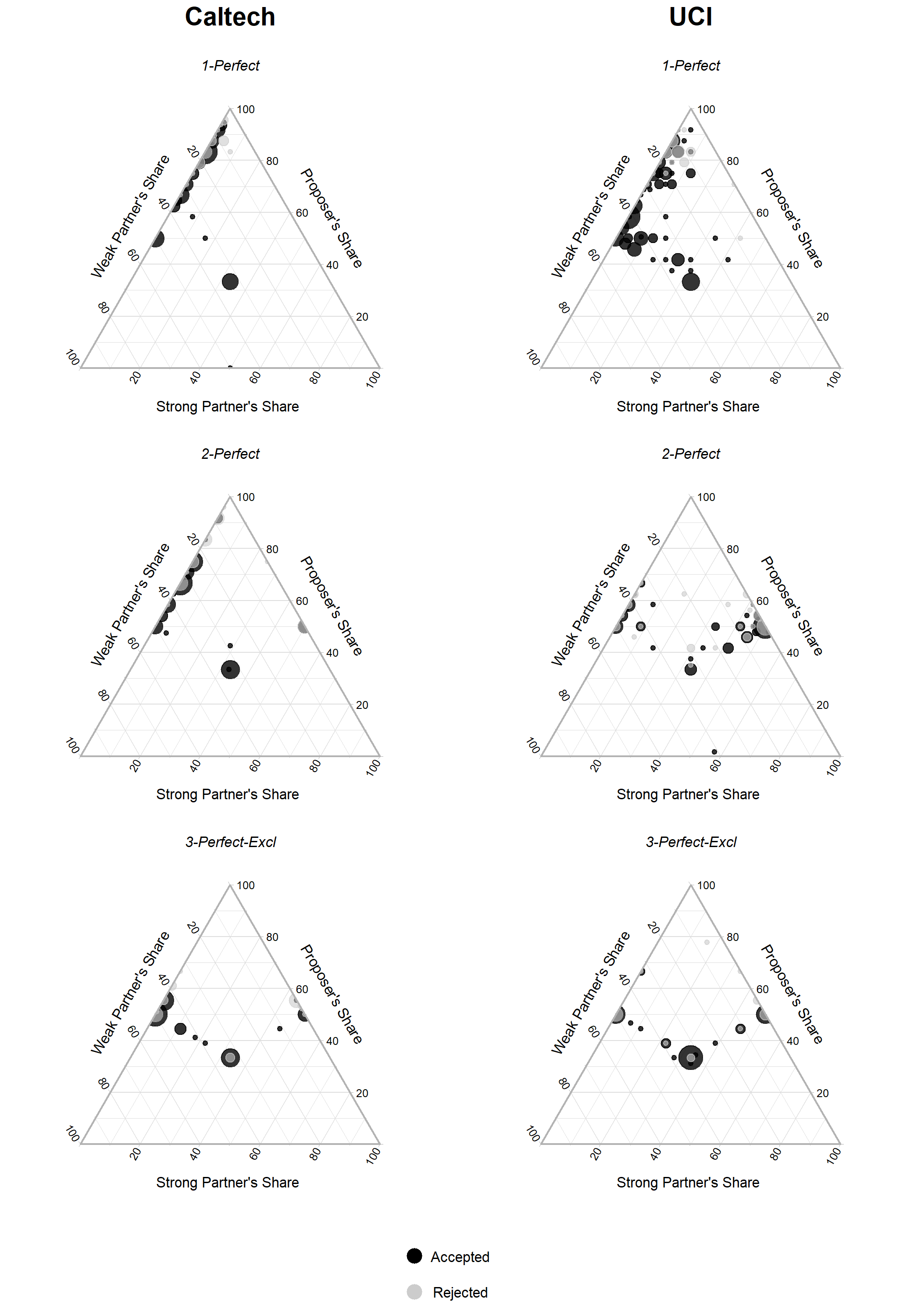}
    \caption{Distributions of Offers and Acceptances}
    \label{fig:acceptedrejected}
    \end{center}
    
    {\footnotesize \underline{Notes:} For the 3-Perfect treatment, we show results for cases in which the first and second proposers differ (3-Perfect-Excl). \Cref{fig:appendixacceptreject} in \Cref{Appendix-AdditionalTables1} shows similar plots for 3-Perfect-Incl, as well as the other three-round treatments at UCI with less than perfect predictability.}
\end{figure}
For Caltech, we see similar patterns for the one-round and two-round treatments. In both cases, there is a cluster of MWC proposals that target the weak partner by offering shares between zero and one-half, and a cluster of egalitarian offers targeting the grand coalition. Rejections are concentrated among the most selfish MWC offers (close to the top vertex). In the two-round treatment, we see some egalitarian MWC offers to the strong partner, who rejects them. Moving from the two-round to the three-round games, we see the migration to egalitarian offers, both MWC and grand coalition, noted previously. We also see the emergence of ``hybrid egalitarian" offers which allocate less than one-third to one of the non-proposers and splits the remainder equally between the proposer and the other player. These are the allocations in the figure that lie between the centroid (equal-split) and the midpoint of either the left or right face of the simplex. Rejections are still more concentrated among most selfish offers, particularly for those that target strong partners, but the pattern is less striking than in the one-round and two-round games.

For UCI, the pattern of offers and responses in one-round games is broadly similar to the one observed at Caltech, although we see more variation in offers, as well as more hybrid egalitarian offers, which are not observed at Caltech in either one- or two-round settings. Rejections are once again concentrated among the most selfish offers, but also among those that target strong partners. For two-round games, the patterns at UCI and Caltech are much less similar. For UCI, we see a high concentration of egalitarian (or nearly egalitarian) MWC offers to strong partners, a smaller number of egalitarian (or nearly egalitarian) MWC offers to weak partners, as well as a smattering of egalitarian or hybrid egalitarian grand coalition offers. Committees reject the most selfish offers more frequently, but it is not obvious that the strong partners remain less accommodating than the weak partners. For three-round games, the UCI and Caltech patterns are once again broadly similar, in that we see concentrations of egalitarian MWC offers, egalitarian grand coalition offers, and hybrid egalitarian offers. Egalitarian grand coalition offers pass with the highest frequency, and committees systematically reject the most selfish proposals. The striking symmetry of the three-round UCI figure suggests that neither proposers nor their potential partners fully understand how to differentiate weak and strong coalition partners.

\Cref{table:LogitNew} presents estimates of logit models relating the binary voting outcome (yes or no) for a non-proposing players to a constant and three variables: (1) an indicator specifying whether the player is weak or strong according to theory ("Strong Partner"), (2) the share proposed for that player ("Own Share"), and (3) the inequality of the offer (``Gini Coefficient'').\footnote{We experimented with more complicated models, including ones with MWC indicators, quadratic terms, and alternative measures of inequality, but found that they did not meaningfully improve the fit or change the results.}
As expected, in every treatment, the likelihood of voting in favor of a proposal is increasing in the player's own share. Moreover, the magnitude of this effect is similar across treatments and locations (Caltech versus UCI).

\begin{table}[t!]
\caption{Determinants of Proposal Acceptance}
\label{table:LogitNew}
\begin{center}\vspace{-.15in}
{\footnotesize \scalebox{0.7}{
\begin{tabular}{@{\extracolsep{5pt}}lD{.}{.}{-3} D{.}{.}{-3} D{.}{.}{-3} D{.}{.}{-3} D{.}{.}{-3} D{.}{.}{-3} D{.}{.}{-3} } 
\\[-1.8ex]\hline 
\hline \\[-1.8ex] 
 & \multicolumn{1}{c}{\shortstack{Caltech:\\1-Perf}} & \multicolumn{1}{c}{\shortstack{Caltech:\\2-Perf}} & \multicolumn{1}{c}{\shortstack{Caltech:\\ 3-Perf-Excl}} & \multicolumn{1}{c}{\shortstack{UCI:\\1-Perf}} & \multicolumn{1}{c}{\shortstack{UCI:\\ 2-Perf}} & \multicolumn{1}{c}{\shortstack{UCI:\\ 3-Part-Incl}} & \multicolumn{1}{c}{\shortstack{UCI:\\ 3-Perf-Excl}} \\ 
\\[-1.8ex] & \multicolumn{1}{c}{(1)} & \multicolumn{1}{c}{(2)} & \multicolumn{1}{c}{(3)} & \multicolumn{1}{c}{(4)} & \multicolumn{1}{c}{(5)} & \multicolumn{1}{c}{(6)} & \multicolumn{1}{c}{(7)}\\ 
\hline \\[-1.8ex] 
 Strong Partner & -2.472^{***} & -2.800^{***} & -1.358^{***} & -2.183^{***} & -1.028^{***} & -0.074 & -0.195 \\ 
  & (0.647) & (0.707) & (0.501) & (0.641) & (0.382) & (0.386) & (0.403) \\ 
  & & & & & & & \\ 
 Own Share & 8.271^{***} & 9.211^{***} & 10.851^{***} & 12.561^{***} & 9.522^{***} & 9.224^{***} & 15.478^{***} \\ 
  & (2.968) & (1.883) & (3.434) & (2.662) & (1.241) & (1.963) & (4.560) \\ 
  & & & & & & & \\ 
 Gini Coefficient & -0.040 & -1.628 & -4.519^{**} & 0.514 & -1.626 & -3.542^{***} & -8.610^{***} \\ 
  & (2.022) & (1.293) & (1.903) & (1.884) & (1.348) & (1.363) & (2.800) \\ 
  & & & & & & & \\ 
 Constant & -0.420 & -0.895 & -2.428^{*} & -1.596 & -2.435^{***} & -2.239^{***} & -3.707^{**} \\ 
  & (1.506) & (0.717) & (1.345) & (1.328) & (0.515) & (0.622) & (1.546) \\ 
  & & & & & & & \\ 
\hline \\[-1.8ex] 
Pseudo R2 & 0.48 & 0.51 & 0.46 & 0.58 & 0.36 & 0.35 & 0.46 \\ 
Observations & \multicolumn{1}{c}{176} & \multicolumn{1}{c}{176} & \multicolumn{1}{c}{156} & \multicolumn{1}{c}{304} & \multicolumn{1}{c}{304} & \multicolumn{1}{c}{210} & \multicolumn{1}{c}{218} \\ 
\hline 
\hline \\[-1.8ex] 
\textit{Note:}  & \multicolumn{7}{r}{$^{*}$p$<$0.1; $^{**}$p$<$0.05; $^{***}$p$<$0.01} \\ 
\end{tabular}
}}
\end{center}
\par
\vspace{2mm} {\footnotesize \underline{Notes:} This table reports estimates of logit models relating the binary voting outcomes (yes or no) for individual non-proposing players to an indicator specifying whether the player is weak or strong according to theory, the share proposed for that player, and a measure of inequality (the Gini coefficient). Standard errors are clustered by session. }
\end{table}

The coefficients of "Strong Partner" speak to the self-awareness hypothesis. At Caltech, theoretically strong partners are much more likely to reject offers than theoretically weak partners in one-round games. This finding indicates a high degree of awareness concerning their own bargaining strength. For two-round games, the corresponding coefficient indicates a similar degree of self-awareness. For three-round games, the magnitude of the coefficient attenuates substantially, indicating a lower level of self-awareness. However, even in three-round games, the coefficient remains negative and statistically significant. Accordingly, Caltech subjects do have some ability to discern their bargaining strengths even in the most complex settings we consider. At UCI, we similarly find that theoretically strong partners are much more likely to reject offers than theoretically weak partners in one-round games. However, in contrast to the Caltech pattern, this tendency is substantially attenuated in two-round games. Surprisingly, UCI voters retain some ability to distinguish their own strength or weakness in the 2-Perfect game even though, as shown above, proposers target strong players, apparently believing them to be weak. Possibly voters feel greater motivation to think through their own prospects than proposers. Significantly, in three-round games, there is no relationship between voting responses and theoretical bargaining strength. These findings imply that UCI subjects reliably self-identify as weak or strong only in the simplest bargaining settings. Their ability to do so declines rapidly with the game's strategic complexity, and vanishes entirely in the most complex settings we consider.

 The pattern of coefficients for the ``Gini Coefficient'' shows that, as subjects become less able to determine their own bargaining strength, they place progressively greater weight on fairness. Strikingly, fairness has no impact on voting responses at either Caltech or UCI in one-round games, where players understand their bargaining strengths. Accordingly, there is no indication that they care about fairness intrinsically in the settings we consider. In two-round games, the coefficients of our inequality measure turn negative (meaning that greater inequality reduces the odds of passage), but the effect lacks statistical significance in both locations. In three-round games, wherein players demonstrably struggle to determine their own bargaining strengths, the corresponding coefficients become more negative and highly statistically significant, especially at UCI. Thus, the tendency to place substantial weight on equity appears to emerge in this setting as a response to strategic complexity (which obscures which players are in weak or strong positions), rather than as an expression of innate preference.

\subsubsection{Optimality Of Proposals}\label{sec:Optimality}

Next, we ask whether proposals are approximately optimal given the voting behavior documented in the previous subsection. In other words, do actual proposals differ from equilibrium proposals because proposers fail to optimize correctly, or because they respond rationally to empirical voting patterns that differ from theoretical predictions?

To evaluate optimization failures, we determine the proposals that would maximize the proposer's expected payoff given the observed mapping from proposals to votes. To make this calculation, we assume that the continuation outcome in the event of rejection does not depend on the terms of the rejected proposal. This assumption is consistent with the implications of backward induction, assuming players vote as if they are pivotal (a common premise) and that indifference is insufficiently prevalent to support multiple equilibria. It also has empirical support, in that continuation outcomes are not systematically related to rejected offers. Under this assumption, for each treatment, we can estimate the first proposer's expected payoff conditional on rejection of her offer by calculating the average payoff for first proposers when voters reject their offers (the \emph{mean rejection payoff}, henceforth MRP). We list these average payoffs in \Cref{table:MRP}.

We use the logit regression models shown in  \Cref{table:LogitNew}, which describe non-proposing players' voting responses to proposals, along with the MRP, to determine the mapping from proposals to expected first-proposer payoffs. Specifically, we use the regression models to compute the probability a proposal passes as a function of its terms. We then multiply the proposer's share by this probability and add the product of the MRP and the estimated rejection probability. \Cref{fig:expectedpayoff} uses contour heat maps to exhibit the first proposers' expected payoff mappings for the one-, two-, and three-round \emph{Perfect} treatments.\footnote{For the 3-Perfect treatment, we show results for cases in which the first and second proposers differ (3-Perfect-Excl). In \Cref{Appendix-AdditionalTables}, we show similar plots for 3-Perfect-Incl, as well as the other three-round treatments at UCI.} In each case, we associate offers with points on a simplex, exactly as in \Cref{fig:acceptedrejected}. 

\begin{figure}[p]
    \begin{center}
    \includegraphics[width=6in]{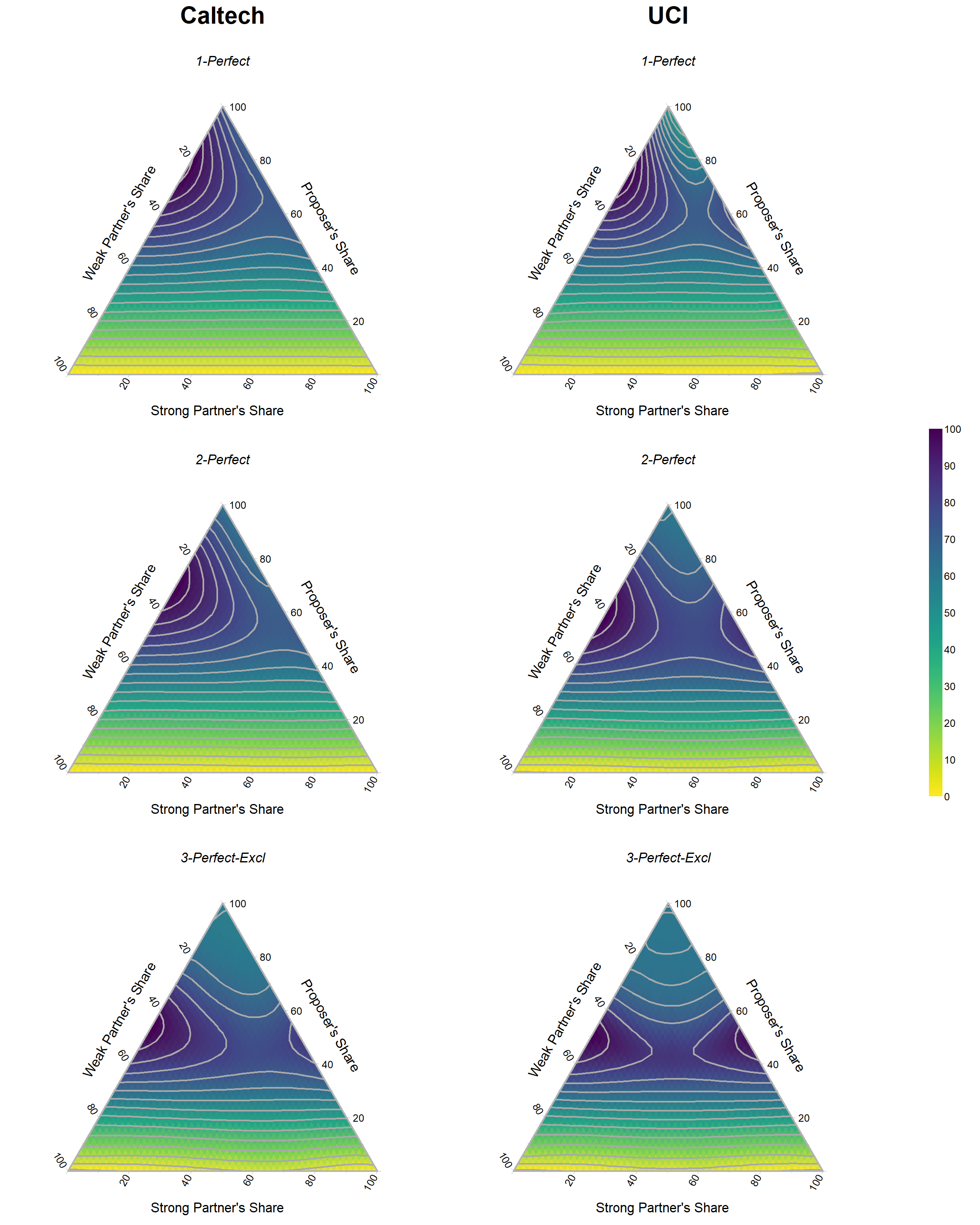}
    \caption{Mappings from First Proposals to Expected Payoffs}
    \label{fig:expectedpayoff}
    \end{center}
    
    {\footnotesize \underline{Notes:} For the 3-Perfect treatment, we show results for cases in which the first and second proposers differ (3-Perfect-Excl). \Cref{fig:appendixheatmap} in \Cref{Appendix-AdditionalTables2} shows similar plots for 3-Perfect-Incl, as well as the other three-round treatments at UCI with less than perfect predictability.}
\end{figure}

In all cases, the offer that maximizes the proposer's expected payoff is an MWC offer. Interestingly, it sometimes delivers an expected payoff below 50\% of the prize, contrary to Theoretical Result 1, because observed voting patterns are less favorable to the proposer than equilibrium voting behavior. At Caltech, the ideal proposal targets the weak player and seeks to retain 76\% of the prize in one-round games, 69\% of the prize in two-round games, and $53\%$ in three-round games (for which there is a weak partner). At UCI, the ideal proposal targets the weak player in one- and two-round games, and seeks to retain 73\% of the prize in the first instance and $59\%$ in the second. For three-round games, the ideal proposal at UCI seeks to form an MWC and retains only 51\% of the prize, but whether it targets the weak player or the strong player is nearly inconsequential. Thus, in both locations,  being egalitarian within an MWC becomes increasingly attractive from a purely selfish perspective as strategic complexity increases. Accordingly, the nature of voting responses by non-proposing players accounts for much of the observed migration to egalitarianism.

Next, we evaluate the quality of first proposers' offers by comparing them to optimized proposals. We perform two versions of this calculation. For the first, we calculate the first proposer's expected payoff based on the average offer for each coalition type and then divide by the optimized expected payoff; to gauge overall performance, we compute the average of these measures weighted by the frequencies of the coalition types. We depict this measure in \Cref{table:optimization}. For the second, we calculate the average ratio between the first proposer's actual payoff and the optimized expected payoff, both in the aggregate and for each coalition type chosen by the first proposer. We provide results based on this measure in \cref{Appendix-Optimality}. The relative merits of these two approaches depend on how one thinks about the observed variation in offers within treatments. If the variation mainly reflects observational noise, the first approach is more relevant, while if it arises mainly from true heterogeneity with respect to preferences or strategic thinking, the second approach is more appropriate.

\begin{table}[t!] \label{table:optimization}
\begin{center}
\begin{tabular}[t]{>{}llccccccc}
\toprule
Location & Treatment & \makecell{Optimal\\Payoff}  & \makecell{MWC\\Weak} & \makecell{MWC\\Strong} & \makecell{MWC\\NA} & \makecell{Grand\\Coalition} & Dictatorial & Agg \\
\midrule
\textbf{Caltech} & 1-Perfect & 64 & 100 & -- & -- & 58 & 80 & 95\\
\textbf{Caltech} & 2-Perfect & 60 & 100 & 72 & -- & 59 & 67 & 92\\
\textbf{Caltech} & 3-Perfect-Excl & 46 & 100 & 82 & -- & 75 & -- & 91\\
\textbf{Caltech} & 3-Perfect-Incl & 55 & -- & -- & 87 & 78 & 80 & 85\\
\textbf{UCI} & 1-Perfect & 65 & 98 & 71 & -- & 73 & 54 & 90\\
\textbf{UCI} & 2-Perfect & 49 & 99 & 87 & -- & 75 & -- & 87\\
\textbf{UCI} & 3-Perfect-Excl & 44 & 100 & 98 & -- & 80 & -- & 90\\
\textbf{UCI} & 3-Perfect-Incl & 47 & -- & -- & 81 & 79 & -- & 79\\
\textbf{UCI} & 3-None & 45 & -- & -- & 76 & 77 & -- & 77\\
\textbf{UCI} & 3-Partial-Excl & 41 & -- & -- & 90 & 84 & -- & 88\\
\textbf{UCI} & 3-Partial-Incl & 42 & 100 & 99 & -- & 79 & 41 & 91\\
\bottomrule
\end{tabular}
\caption{Proposer Optimization Rates}
\end{center}     
    {\footnotesize \underline{Notes:} This table displays the optimization rates of the first proposer using our first measure, which captures how well the average offer within a coalition type performs relative to the empirically optimal payoffs, reporting this ratio as a percentage. The final ``Agg'' column reflects an average of the performance across coalition types weighted by the frequency of each coalition type.}
\end{table}

Several lessons emerge from this analysis. First, proposers perform reasonably well: across all protocols and locations, they achieve 77\% to 95\% of optimized payoffs according to the first measure. Table 5 provides a detailed breakdown of these optimization rates.\footnote {The results using the second measure are nearly identical, with optimized payoff rates ranging from 80\% to 92\%. See \Cref{table:measure2} in \Cref{Appendix-Optimality} for details.} To put these ranges into perspective, in settings where the first proposer is also not the second proposer, our calculations imply that first proposers would achieve between 8\% and 61\% of the optimized payoff by making a rejected proposal, between 40\% and 80\% if they made dictatorial offers, and between 50\% and 80\% if they offered an equal split within a grand coalition.\footnote{See \Cref{table:GCEgalitarian,tb:MRPoptim} in \Cref{Appendix-Optimality}.} Relative to these measures, the performance of proposers indicates a meaningful degree of sophistication.

Second, Caltech subjects perform slightly better than UCI subjects: averaged across the four treatments, the Caltech subjects achieved 91\% of the optimized payoff according to the first measure and 88\% according to the second; for UCI subjects, the corresponding figures are 88\% and 84\% respectively. This comparison is consistent with other evidence of greater sophistication on the part of Caltech subjects discussed elsewhere in this paper. 

Third, we also see no clear relationship between complexity and performance. In other words, these measures of performance do not decline systematically as we move from one-round to two-round to three-round games. Apparently, the tendency for voting players to evaluate proposals based on simple egalitarian criteria, along with proposers' intuitive understanding of those criteria, offsets the greater strategic complexity of the environment.

Finally, according to the first measure, the average offer that seeks to form an MWC with a weak player always comes close to achieving expected-payoff maximization. This finding is of particular interest because, in discussing \Cref{fig:MWCegalitarian}, we noted the sharp progression toward egalitarianism among MWC offers targeting weak players that coincides with increases in strategic complexity. As it turns out, despite being contrary to equilibrium implications, this pattern is nearly optimal for the proposers.\footnote{We noted in \Cref{sec:Proposals} that in two and three round games, first proposers at UCI frequently target the strong partner rather than the weak one, which is sub-optimal given the predicted voting behavior. That said, according to either measure, these offers do not reduce their payoffs substantially: in the two-round games, such proposals accrue 86\%-87\% of the optimized payoff, and more than 97\% in the three-round game. }

\section{Conclusion}\label{Section-Conclusion}

In this study, we examined legislative bargaining protocols experimentally to determine (1) the extent to which people identify and exploit weak coalition partners, (2) the reasons for failures to do so in settings where they occur, and (3) the systematic behavioral patterns that arise in those settings. Our experiment focuses on the role of strategic complexity by varying the number of bargaining rounds from one to three in a setting with asymmetrically positioned players, which makes strength and weakness progressively more difficult to ascertain. We
employed two distinct subject pools, one of which (Caltech) displays greater
quantitative aptitude and skill than the other (UCI). Our primary finding is that, regardless of analytic sophistication, players gravitate toward egalitarianism as strategic complexity progressively obscures relative negotiating strengths. Analytic sophistication only affects the speed with which this migration occurs. More specifically, we reach the following seven main conclusions.

First, in all settings, the first proposers receive far smaller shares than predicted by standard bargaining theory, and this discrepancy increases with the game's degree of strategic complexity. 
For three-round games, theoretical predictions concerning comparative statics (involving variations that ought to affect the identity and exploitability of the weak partner) fail entirely.

Second, as the number of rounds increases, there is a dramatic increase in the frequency with which the selected allocations involve equal division, either within minimum winning coalitions or grand coalitions. 

Third, results for relatively simple (one- and two- round) settings confirm that subjects nevertheless have exploitative intent, in that they make MWC offers with high frequency, and seek to extract relatively high shares. 

Fourth, as strategic complexity increases, exploitative ability declines: proposers have greater difficulty identifying weak partners, and increasingly adhere to within-coalition egalitarianism. Even the non-egalitarian offers exhibit a hybrid form of egalitarianism---i.e., a small share for one player and an equal split between the other two. A comparison of results for Caltech and UCI suggests that analytic sophistication dampens but does not eliminate these patterns.

Fifth, increasing strategic complexity reduces self-awareness of bargaining power among non-proposers. In one-round games, strong partners are significantly less accommodating than weak players. This differential attenuates substantially as strategic complexity increases at Caltech, and vanishes entirely at UCI. 

Sixth, in the simplest settings, non-proposing players place no weight on fairness when casting their votes, but fairness begins to play a much larger role in their evaluations once strategic complexity obscures their own bargaining power. 

Seventh, the migration toward egalitarian that accompanies rising complexity (our second main result) is, in substantial part, a rational response to non-proposers' increasing emphasis on fairness (our sixth main result). Most notably, egalitarian offers to MWCs become nearly optimal.

Taken together, our findings suggest a new rationale for fairness in legislative bargaining: when it is difficult to discern the strength of others' bargaining positions, a proposer may offer a fair split (within a coalition type) to hedge against her own strategic uncertainty or that of others. The proposer's motivation is fundamentally strategic---she offers a fair split not from a desire for fairness or to appear fair, but because such a division is likely to pass, even if she cannot identify weak players or weak players mistakenly believe they are strong. An important direction for future research is to formalize this strategic theory, building on the framework in \Cref{Subsection-egalitarianism}, and to better understand its implications for organizations and politics.

\begin{singlespace}
{\small
	\addcontentsline{toc}{section}{References}
	\setlength{\bibsep}{.25\baselineskip}
	\bibliographystyle{aer}
	\bibliography{AABP}

@article{simon1956rational,
  title={Rational choice and the structure of the environment.},
  author={Simon, Herbert A},
  journal={Psychological review},
  volume={63},
  number={2},
  pages={129},
  year={1956},
  publisher={American Psychological Association}
}

@article{crawford2013survey,
Author = {Crawford, Vincent P. and Costa-Gomes, Miguel A. and Iriberri, Nagore},
Title = {Structural Models of Nonequilibrium Strategic Thinking: Theory, Evidence, and Applications},
Journal = {Journal of Economic Literature},
Volume = {51},
Number = {1},
Year = {2013},
Month = {March},
Pages = {5–62},
DOI = {10.1257/jel.51.1.5},
URL = {https://www.aeaweb.org/articles?id=10.1257/jel.51.1.5}}

@article{nagel1995unraveling,
  title={Unraveling in guessing games: An experimental study},
  author={Nagel, Rosemarie},
  journal={The American economic review},
  volume={85},
  number={5},
  pages={1313--1326},
  year={1995},
  publisher={JSTOR}
}

@unpublished{guanoprea2024faces,
    author = {Menglong Guan and Ryan Oprea},
    title = {Three Faces of Complexity in Strategic Choice},
    note = {Working Paper},
    year = {2024}
}

@article{neyman1985bounded,
  title={Bounded complexity justifies cooperation in the finitely repeated prisoners' dilemma},
  author={Neyman, Abraham},
  journal={Economics letters},
  volume={19},
  number={3},
  pages={227--229},
  year={1985},
  publisher={Elsevier}
}

@article{stahl1995players,
  title={On players' models of other players: Theory and experimental evidence},
  author={Stahl, Dale O. and Wilson, Paul W.},
  journal={Games and Economic Behavior},
  volume={10},
  number={1},
  pages={218--254},
  year={1995},
  publisher={Elsevier}
}

@article{camerer2004cognitive,
  title={A cognitive hierarchy model of games},
  author={Camerer, Colin F and Ho, Teck-Hua and Chong, Juin-Kuan},
  journal={The Quarterly Journal of Economics},
  volume={119},
  number={3},
  pages={861--898},
  year={2004},
  publisher={MIT Press}
}

@incollection{OpreaSurvey, 
	author = {Oprea, Ryan},
	editor = {XXX}, 
	publisher = {XXX},
	title = {Complexity and Its Measurement},
    booktitle = {Handbook of XXXX},
	volume = {1},
	year = {2025}}

@incollection{Kalai1990, 
	author = {Kalai, Ehud},
	editor = {Ichiishi, Tatsuro and Neyman, Abraham and Tauman, Yair}, 
	publisher = {Academic Press},
	title = {Bounded Rationality and Strategic Complexity in Repeated Games},
    booktitle = {Game Theory and Applications},
	year = {1990}}

@article{Andreoni2020, 
	author = {Andreoni, James and Aydin, Deniz and Barton, Blake and Bernheim, B. Douglas and Naecker, Jeffrey},
	journal = {Journal of Political Economy},
	title = {When Fair Isn't Fair: Understanding Choice Reversals Involving Social Preferences},
    volume = {128},
    number = {5},
	pages = {1673--1711},
	year = {2020}}

@article{Andreoni2009, 
	author = {Andreoni, James and Bernheim, B. Douglas},
	journal = {Econometrica},
	title = {Social Image and the 50-50 Norm: A Theoretical and Experimental Analysis of Audience Effects},
    volume = {77},
    number = {5},
	pages = {1607--1636},
	year = {2009}}

@article{Rubinstein1988, 
	author = {Abreu, Dilip and Rubinstein, Ariel},
	journal = {Econometrica},
	title = {The structure of Nash equilibria in repeated games with finite automata},
    volume = {56},
	pages = {1259--1282},
	year = {1988}}

@article{oprea2024decisions,
  title={Decisions under risk are decisions under complexity},
  author={Oprea, Ryan},
  journal={American Economic Review},
  volume={114},
  number={12},
  pages={3789--3811},
  year={2024},
  publisher={American Economic Association 2014 Broadway, Suite 305, Nashville, TN 37203}
}

@article{puri2025simplicity,
  title={Simplicity and risk},
  author={Puri, Indira},
  journal={The Journal of Finance},
  volume={80},
  number={2},
  pages={1029--1080},
  year={2025},
  publisher={Wiley Online Library}
}

@unpublished{babohrenimas,
    author = {Cuimin Ba and J. Aislinn Bohren and Alex Imas},
    title = {Over- and Underreaction to Information},
    note = {Working Paper},
    year = {2024}
}

@article{alaoui2016endogenous,
  title={Endogenous depth of reasoning},
  author={Alaoui, Larbi and Penta, Antonio},
  journal={The Review of Economic Studies},
  volume={83},
  number={4},
  pages={1297--1333},
  year={2016},
  publisher={Oxford University Press}
}

@article{bernheim2020empirical,
  title={On the empirical validity of cumulative prospect theory: Experimental evidence of rank-independent probability weighting},
  author={Bernheim, B Douglas and Sprenger, Charles},
  journal={Econometrica},
  volume={88},
  number={4},
  pages={1363--1409},
  year={2020},
  publisher={Wiley Online Library}
}

@article{enke2023cognitive,
  title={Cognitive uncertainty},
  author={Enke, Benjamin and Graeber, Thomas},
  journal={The Quarterly Journal of Economics},
  volume={138},
  number={4},
  pages={2021--2067},
  year={2023},
  publisher={Oxford University Press}
}

@incollection{AgranovSurvey,
	author = {Agranov, Marina},
	editor = {Emin Karagozoglu and Kyle Hyndman}, 
	publisher = {Palgrave Macmillan},
	title = {Legislative Bargaining Experiments},
    booktitle = {Bargaining: Current Research and Future Directions},
	year = {2021}}

@incollection{Chatterjee2009,
	author = {Chatterjee, Kalyan and Sabourian, Hamid},
	editor = {Meyers, R}, 
	publisher = {Springer},
	title = {Game Theory and Strategic Complexity},
    booktitle = {Encyclopedia of Complexity and Systems Science},
	year = {2009}}

@article{ali2019predictability,
	author = {Ali, S. Nageeb and Bernheim, B. Douglas and Fan, Xiaochen},
	journal = {The Review of Economic Studies},
	number = {2},
	pages = {500--525},
	publisher = {Oxford University Press},
	title = {Predictability and power in legislative bargaining},
	volume = {86},
	year = {2019}}

@incollection{palfrey2016experiments,
	address = {Princeton, N.J.},
	author = {Palfrey, Thomas R},
	editor = {John Kagel and Alvin E. Roth},
	pages = {347--434},
	publisher = {Princeton University Press},
	title = {Experiments in political economy},
    booktitle = {The Handbook of Experimental Economics},
	volume = {2},
	year = {2016}}

@article{binmore1985,
	author = {Binmore, Ken and Shaked, Avner and Sutton, John},
	journal = {American Economic Review},
	pages = {1178--1180},
	title = {Testing Noncooperative Bargaining Theory: A Preliminary Study},
	volume = {75},
    number = {5},
	year = {1985}}

@article{binmore1991,
	author = {Binmore, Ken and Morgan, Peter and Shaked, Avner and Sutton, John},
	journal = {Games and Economic Behavior},
	pages = {295--322},
	title = {Do People Exploit Their Bargaining Power? An Experimental Study},
	volume = {3},
	year = {1991}}

@article{neelin1988,
	author = {Neelin, Janet and Sonnenschein, Hugo and Spiegel, Matthew},
	journal = {American Economic Review},
	pages = {824--836},
	title = {A Further Test of Noncooperative Bargaining Theory: Comment},
	volume = {78},
    number = {4},
	year = {1988}}

@article{fehr-charness,
    author = {Ernst Fehr and Gary Charness},
    title = {Social Preferences: Fundamental Characteristics and Economic Consequences},
    journal = {Journal of Economic Literature},
    year = {2024}
}

@article{spiegel1994,
	author = {Spiegel, Matthew and Currie, Janet and Sonnenschein, Hugo and Sen, Arunava},
	journal = {Games and Economic Behavior},
	pages = {104--115},
	title = {Understanding When Agents are Fairmen or Gamesmen},
	volume = {7},
	year = {1994}}

@article{miller2018legislative,
	author = {Miller, Luis and Montero, Maria and Christoph Vanberg},
	journal = {Games and Economic Behavior},
	pages = {60--92},
	publisher = {Elsevier},
	title = {Legislative bargaining with heterogeneous disagreement values: Theory and experiments},
	volume = {107},
	year = {2018}}

@article{knez1995outside,
	author = {Knez, Marc J and Camerer, Colin F},
	journal = {Games and Economic Behavior},
	number = {1},
	pages = {65--94},
	publisher = {Elsevier},
	title = {Outside options and social comparison in three-player ultimatum game experiments},
	volume = {10},
	year = {1995}}

@article{guth1998information,
	author = {G{\"u}th, Werner and Van Damme, Eric},
	journal = {Journal of mathematical Psychology},
	number = {2-3},
	pages = {227--247},
	publisher = {Elsevier},
	title = {Information, strategic behavior, and fairness in ultimatum bargaining: An experimental study},
	volume = {42},
	year = {1998}}

@article{frechette2005nominal,
	author = {Frechette, Guillaume and Kagel, John H and Morelli, Massimo},
	journal = {Journal of Public Economics},
	number = {8},
	pages = {1497--1517},
	publisher = {Elsevier},
	title = {Nominal bargaining power, selection protocol, and discounting in legislative bargaining},
	volume = {89},
	year = {2005}}

@article{li2024designing,
  title={Designing Simple Mechanisms},
  author={Li, Shengwu},
  journal={Journal of Economic Perspectives},
  volume={38},
  number={4},
  pages={175--192},
  year={2024},
  publisher={American Economic Association 2014 Broadway, Suite 305, Nashville, TN 37203-2418}
}

@incollection{guthtietz1988,
	author = {G{\"u}th, Werner and Tietz, Reinhard},
	booktitle = {Bounded Rational Behavior in Experimental Games and Markets},
	editor = {Reinhard Tietz and Wolf Albers and Reinhard Selten},
	pages = {111--128},
	publisher = {Springer-Verlag},
	title = {Ultimatum Bargaining for a Shrinking Cake: An Experimental Analysis},
	year = {1988}}

@article{diermeier2006self,
	author = {Diermeier, Daniel and Gailmard, Sean},
	journal = {Quarterly Journal of Political Science},
	number = {4},
	pages = {327--350},
	title = {Self-interest, inequality, and entitlement in majoritarian decision-making},
	volume = {1},
	year = {2006}}

@article{agranov2014communication,
	author = {Agranov, Marina and Tergiman, Chloe},
	journal = {Journal of Public Economics},
	pages = {75--85},
	publisher = {Elsevier},
	title = {Communication in multilateral bargaining},
	volume = {118},
	year = {2014}}

@book{billingsley1995,
	author = {Billingsley, Patrick},
	edition = {Third Edition},
	publisher = {John Wiley \& Sons},
	title = {Probability and Measure},
	year = {1995}}

@book{billingsley1999,
	author = {Billingsley, Patrick},
	edition = {Second Edition},
	publisher = {John Wiley \& Sons},
	title = {Convergence of Probability Measures},
	year = {1999}}

@book{santambrogio2023,
  title={A Course in the Calculus of Variations:
Optimization, Regularity, and Modeling},
  author={Santambrogio, Filippo},
  year={2023},
  publisher={Springer}
}

@article{baron1989bargaining,
	author = {Baron, David P. and Ferejohn, John A.},
	journal = {American Political Science Review},
	number = {4},
	pages = {1181--1206},
	publisher = {JSTOR},
	title = {Bargaining in legislatures},
	volume = {83},
	year = {1989}}
}
\end{singlespace}

\newpage
\appendix

\section{Appendix: Proof of \Cref{Proposition-Egalitarian}}\label{Appendix-Theory}

We begin with notation used in the argument below. Let $G$ denote a generic bivariate CDF on $\Re^2$ and $d$ denote a generic expected payoff for the proposer in the continuation game following rejection of her offer. We restrict attention to distributions $G$ such that $G(x,x)>0$ for some  $x<(1-d)/2$; we later show that this case is relevant for our analysis as our limit belief satisfies this property (\Cref{Assumption-Thresholds}) and we consider a sequence of beliefs that converge  to this limit belief (in the weak topology). 

Let $s\equiv (s_A,s_B,s_C)$ denote a generic division of the dollar and $S$ be the set of all divisions. The probability that offer $s=(s_A,s_B,s_C)$ is accepted by at least one non-proposing player is \begin{align*}
	    \Lambda(s_B,s_C,G)&\equiv \int_{\tau_B\times \tau_C} \mathbb{1}_{\max\{s_B-\tau_B,s_C-\tau_C\}\geq 0}\,dG.
\end{align*}
 This term is the probability of the event $\left[(-\infty,s_B]\times \Re\right]\cup\left[\Re\times (-\infty,s_C]\right]$. We denote the proposer's expected payoff from offer $s$ by  $\pi(s,G,d)\equiv d+(s_A-d)\Lambda(s_B,s_C,G)$.\medskip

We first argue that the proposer's optimal offer gives her at least her disagreement payoff $d$.
\begin{lemma}\label{Lemma-BoundingOfferSpace}
    For every $s$ in which $s_A\leq d$, there exists $s'$ in which $s'_A> d$ such that $\pi(s',G,d)>\pi(s,G,d)$.
\end{lemma} 
\begin{proof}
Observe that $\pi(s,G,d)\leq d$. Consider an offer $s'$ in which $s'_B=s'_C=x\in (0,\frac{1-d}{2})$ and $G(x,x)>0$. Such an $x$ exists given the restriction in the opening paragraph. By construction, $s'_A>d$. Moreover, $\Lambda(x,x,G)\geq G(x,x)>0$, where the  first inequality follows from $(-\infty,s_B]\times (-\infty,s_C]$ being a subset of $\left[(-\infty,s_B]\times \Re\right]\cup\left[\Re\times (-\infty,s_C]\right]$. Therefore, $\pi(s',G,d)>d$.
\end{proof}

We now consider the limit case $(H,d)$, where $H$ is a joint distribution on $[0,\overline\tau]\times[0,\overline\tau]$ where $\overline\tau \geq 1/2$ that has uniform marginals and satisfies \Cref{Assumption-Thresholds}. We use $H_{\tau_i}$ to denote the marginal CDF for $\tau_i$.

\begin{lemma}\label{Lemma-MWCOffers}
    Under distribution $H$, every optimal offer is an MWC offer in which the highest share offered to a non-proposing player, $s_B$, is no more than $\overline\tau$.
\end{lemma}
\begin{proof}
    \Cref{Lemma-BoundingOfferSpace} establishes that every offer $s$ in which $s_A\leq d$ is suboptimal. Therefore, it suffices to consider only offers $s$ in which $s_A>d$. 
    
    We show that every offer $s$ in which $s_B+s_C>\overline\tau$ is suboptimal.\footnote{This case is relevant only if $d+\overline\tau< 1$.} Consider an alternative offer $s'$ in which $s_A'=1-\overline\tau$, $s_B=\overline\tau$, and $s_C=0$. Observe that $s_A'>s_A>d$ and $\Lambda(s_B,s_C,H)\leq 1 = \Lambda(s'_B,s'_C,H)$. Therefore, $\pi(s',H,d)>\pi(s,H,d)$. 
   
    We now consider an offer $s$ in which $s_B+s_C\leq\overline\tau$ and $s_B,s_C$ are each strictly positive. Let $s'$ be the offer in which $s'_A=s_A$, $s'_B = s_B+s_C$, and $s'_C=0$. Observe that
    \begin{align*}
        \Lambda(s_B,s_C)&=H_{\tau_B}(s_B)+H_{\tau_C}(s_C)-H(s_B,s_C)\\
        &=\frac{s_B}{\overline\tau}+\frac{s_C}{\overline\tau}-H(s_B,s_C)\\
        &\geq \frac{s_B}{\overline\tau}+\frac{s_C}{\overline\tau}-H(\min\{s_B,s_C\},\min\{s_B,s_C\})\\
        &>\frac{s_B}{\overline\tau}+\frac{s_C}{\overline\tau}\\
        &= \Lambda(s'_B,s'_C),
    \end{align*}
  where the first equality uses the inclusion-exclusion principle; the second equality uses the fact that $H$ has uniform marginals on $[0,\overline\tau]$; the first inequality uses the fact that $H$ is non-decreasing in each argument; the second inequality uses the fact that $\min\{s_B,s_C\}>0$ and therefore \Cref{Assumption-Thresholds} implies that $H(\min\{s_B,s_C\},\min\{s_B,s_C\})>0$, and the final equality follows by construction. 
  As $s'_A=s_A>d$, we reach the conclusion that $\pi(s',H,d)>\pi(s,H,d)$. 
\end{proof}
We now characterize the optimal MWC offer in the limit case $(H,d)$. 
\begin{lemma}
    The optimal offer in the bargaining problem $(H,d)$ is an MWC offer of the form $\left(\frac{1+d}{2},\frac{1-d}{2},0\right)$.
\end{lemma}
\begin{proof}
  \Cref{Lemma-MWCOffers} establishes that every optimal offer is an MWC offer in which $s_B\leq \overline\tau$. Therefore,
  \begin{align*}
      \argmax_{s\in S}\pi(s,H,d)&=\argmax_{x\in [1-d,\overline\tau]}\left\{(1-x-d) \Lambda(x,0,H) +d\right\}\\
      &=\argmax_{x\in [1-d,\overline\tau]} (1-x-d)(x/\overline\tau).
  \end{align*}
  The optimum to this strictly concave problem is characterized by its first-order condition, which delivers the optimal $x=(1-d)/2$. 
\end{proof}

Having considered the limit case, we establish some results for an arbitrary bivariate CDF $G$ such that $G(x,x)>0$ for some $x< (1-d)/2$. In light of \Cref{Lemma-BoundingOfferSpace}, an offer optimizes $\pi(\cdot,G,d)$ if and only if it optimizes the modified payoff function $\hat\pi(s,G,d)\equiv d+\min\{s_A-d,0\}\Lambda(s,G)$. We argue that $\hat\pi$ is upper semicontinuous in $s$. 
\begin{lemma}\label{Lemma-USC}
    The payoff function $\hat\pi(s,G,d)$ is upper semicontinuous in $s$.
\end{lemma}
\begin{proof}
    Recall that the product of two nonnegative upper semicontinuous functions is upper semicontinuous. Observe that $\min\{s_A-d,0\}$ is nonnegative and continuous in $s$, and $\Lambda(s_B,s_C,G)$ is, by definition, nonnegative. We establish that $\Lambda(s_B,s_C,G)$ is upper semicontinuous in $s$ by first showing that $\Lambda$ is right-continuous in its first two arguments; to minimize notation, we suppress $G$ in the argument below.

    We establish right-continuity in $s_B$ as a symmetric argument applies for $s_C$. Consider a sequence of sets $A_n$ where $A_n \equiv \left[(-\infty,s_B+\frac{1}{n}]\times \Re\right]\cup\left[\Re\times (-\infty,s_C]\right]$. Observe that $A_1\supseteq A_2 \supseteq \ldots$, and $\lim_{n\rightarrow\infty} A_n = \left[(-\infty,s_B]\times \Re\right]\cup\left[\Re\times (-\infty,s_C]\right]$. Therefore, by Theorem 2.1(ii) of \cite{billingsley1995}, $\Lambda(s_B,s_C)=\lim_{n\rightarrow\infty} \Lambda\left(s_B+\frac{1}{n},s_C\right)$, which establishes that $\Lambda$ is right-continuous in its first argument. 

    We now prove that $\Lambda$ is upper semicontinuous in $(s_B,s_C)$. Fix $\epsilon>0$. By right continuity of $\Lambda$ in its first argument, there exists $\delta_B>0$ such that $$\Lambda(s_B+\delta_B,s_C)-\Lambda(s_B,s_C)<\epsilon/2.$$ Analogously, there exists $\delta_C>0$ such that $$\Lambda(s_B+\delta_B,s_C+\delta_C)-\Lambda(s_B+\delta_B,s_C)<\epsilon/2.$$ Therefore, $\Lambda(s_B+\delta_B,s_C+\delta_C)-\Lambda(s_B,s_C)<\epsilon$. Because $\Lambda$ is non-decreasing in each argument, we obtain that for any $(s'_B,s'_C)\leq (s_B+\delta_B,s_C+\delta_C)$, $\Lambda(s_B',s_C')<\Lambda(s_B,s_C)+\epsilon$. Thus, there exists a $\delta$-ball around $(s_B,s_C)$ such that for all $(s'_B,s'_C)$ in that $\delta$-ball, $|\Lambda(s'_B,s'_C)-\Lambda(s_B,s_C)|<\epsilon$.
\end{proof}

\begin{lemma}\label{Lemma-Existence}
    There exists an offer that maximizes $\pi(\cdot,G,d)$.
\end{lemma}
\begin{proof}
    Because $\hat\pi(\cdot,G,d)$ is upper semicontinuous and the unit simplex $S$ is compact, there exists an offer $s$ that maximizes $\hat\pi(\cdot,G,d)$ over $S$. By \Cref{Lemma-BoundingOfferSpace}, any such offer also maximizes $\pi(\cdot,G,d)$. 
\end{proof}

Finally, we complete the argument by establishing that the optimal solution close to this limit case approaches the optimal offer of the limit case. For a bargaining problem $(G,d)$, let $s^*(G,d)\equiv \argmax_{s\in {S}}\hat\pi(s,G,d)$ be the correspondence of optimal proposals and $\pi^*(G,d)\equiv \max_{s\in \overline{S}}\hat\pi(s,G,d)$ be its value. 
\begin{lemma}
    The value function $\pi^*(\cdot)$ is continuous at $(H,d)$ and the correspondence of optimal proposals $s^*(\cdot)$ is upper hemicontinuous at $(H,d)$. 
\end{lemma}
\begin{proof}
    Consider a sequence $(H_n,d_n)$ where $H_n\rightarrow_W H$ and $d_n\rightarrow d$.\medskip
    
\noindent\underline{Step 1}: Fix a small $\epsilon>0$. Because $d_n\rightarrow d$, there 
exists $N_d$ such that for all $n\geq N_d$, $d_n<d+\epsilon$. \Cref{Assumption-Thresholds} assures that for every $x \in \left(0, \frac{1-d-\epsilon}{2}\right)$, $H(x,x)>0$. Therefore, the open set $O\equiv (-\infty,\frac{1-d-\epsilon}
{2})\times (-\infty,\frac{1-d-\epsilon}{2})$ is of strictly positive $H$-measure. The Portmanteau Theorem \citep[Theorem 2.1]{billingsley1999} then implies that there exists $N_H$ such that for every $n\geq N_H$, $Pr_{H_n}(O)$ is strictly positive. Now consider $n\geq \max\{N_d,N_H\}$. For any such bargaining problem $(H_n,d_n)$, there exists $x<\frac{1-d-\epsilon}{2}<\frac{1-d_n}{2}$ such that $H_n(x,x)>0$. Therefore, the assumption that there exists $x<\frac{1-d_n}{2}$ such that $G(x,x)>0$ applies to all $G=H_n$ for $n\geq \max\{N_d,N_H\}$. We therefore restrict attention to such distributions and invoke \Cref{Lemma-BoundingOfferSpace,Lemma-USC,Lemma-Existence} going forward.  

\medskip

\noindent\underline{Step 2}: We show that for every $s\in S$, and for every sequence $(s^n)_{n=1,2,\ldots}$ where $s^n\rightarrow s$, $\limsup_n \hat\pi(s^n,H_n,d_n)\leq \hat\pi(s,H,d)$. To establish this claim, observe that 
    \begin{align*}
        &\limsup_n \hat\pi(s^n,H_n)-\hat\pi(s,H)\\\leq &\limsup_n [\hat\pi(s^n,H_n)- \hat\pi(s,H_n)]+\limsup_n[\hat\pi(s^n,H_n,d_n)-\hat\pi(s^n,H_n,d)]\\&+\limsup_n [\hat\pi(s,H_n)- \hat\pi(s,H)].
    \end{align*}
The first term on the RHS is weakly negative because $\hat\pi(\cdot,H_n)$ is upper semicontinuous by \Cref{Lemma-USC}; the second term is zero because $\hat\pi$ is continuous in its third argument, and the final term on the RHS is weakly negative by the Portmanteau Theorem because $\hat\pi$ is upper semicontinuous and bounded from above.\medskip 

\noindent\underline{Step 3}: We show that for every $\epsilon>0$, for every $s\in S$, there exists a sequence $(s^n)_{n=1,2,\ldots}$ such that $\liminf_{n} \pi(s^n,H_n,d_n)\geq \pi(s,H,d)-\epsilon$.

For $s$ in which $s_A\leq d$, we obtain that $\hat\pi(s,H,d)=d$. This step then holds trivially as it suffices to take any $s^n$ in which $s^n_A=0$.

Now consider $s_A>d$. Let $s^\epsilon=(s_A-2\epsilon,s_B+\epsilon,s_C+\epsilon)$ for $\epsilon\in (0,\frac{s_A-d}{2})$. Also let
\begin{align*}
	    \Lambda^o(s_B,s_C,H)&\equiv \int_{\tau_B\times \tau_C} \mathbb{1}_{\max\{s_B-\tau_B,s_C-\tau_C\}> 0}\,dH,
\end{align*}
denote the probability of the open set $\left[(-\infty,s_B)\times \Re\right]\cup\left[\Re\times (-\infty,s_C)\right]$ under measure $H$. Because $\Lambda^o\left(s_B+\epsilon,s_C+\epsilon,H^n\right)\leq \Lambda\left(s_B+\epsilon,s_C+\epsilon,H^n\right)$,
\begin{align*}
	\hat\pi(s^\epsilon,H_n,d_n)&= d_n+\max\{s_A-\epsilon-d_n,0\}\Lambda\left(s_B+\epsilon,s_C+\epsilon,H_n\right)\\
	&\geq d_n+\max\{s_A-\epsilon-d_n,0\}\Lambda^o\left(s_B+\epsilon,s_C+\epsilon,H_n\right).
\end{align*}
 Given that $H_n\rightarrow_W H$ and $\left[(-\infty,s_B)\times \Re\right]\cup\left[\Re\times (-\infty,s_C)\right]$ is an open set, the Portmanteau Theorem implies that 
 \begin{align}\label{Inequality-PortmanteauOpen}
 \liminf_n \Lambda^o(s_B+\epsilon,s_C+\epsilon,H^n)\geq \Lambda^o(s_B+\epsilon,s_3+\epsilon,H).
 \end{align}
Therefore,, 
\begin{align*}
\liminf_n	\hat\pi(s^\epsilon,H_n,d_n)&\geq d+ \max\{s_A-\epsilon-d,0\}\Lambda^o\left(s_B+\epsilon,s_C+\epsilon,H\right)\\
&\geq d+\max\{s_A-\epsilon-d,0\}\Lambda(s_B,s_C,H)\\
&\geq \hat\pi(s,H)-\epsilon,
\end{align*}
where the first inequality follows from $d_n\rightarrow d$ and \eqref{Inequality-PortmanteauOpen}, the second inequality follows from $\Lambda^o\left(s_B+\epsilon,s_C+\epsilon,H\right)\geq \Lambda(s_B,s_C,H)$, and the final inequality follows from $\Lambda(s_B,s_C,H)\leq 1$.\medskip

\noindent\underline{Step 4}: The desired conclusion then follows from \cite[Propositions 7.4-7.5]{santambrogio2023}: $\hat\pi(\cdot,H_n,d_n)$ then $\Gamma$-converges to $ \hat\pi(\cdot,H,d)$, which then implies that $\pi^*(H_n,d_n)$ converges to $\pi^*(H,d)$ and the correspondence $s^*(\cdot)$ is upper hemicontinuous at $(H,d)$.  
\end{proof}

\newpage

\section*{Online Appendices}

Our Online Appendices are organized as follows:
\begin{itemize}
    \item \Cref{Appendix-TablesBar} shows the tables supporting the bar graphs.
    \item \Cref{Appendix-AdditionalTables} shows additional tables and figures.
    \item \Cref{Appendix-AllMatches} shows exhibits for all matches.
    \item \Cref{Appendix-Instructions} contains some sample instructions.
\end{itemize}

\pagebreak
\section{Tables Supporting Bar Graphs}\label{Appendix-TablesBar}

\begin{table}[!h] 
\centering
\caption{Table For \Cref{fig:FirstProposer}}
\label{table:fig1-realized-share-first-proposer}
\centering
\begin{tabular}[t]{llcc}
\toprule
Location & Treatment & Mean Payoff (\%) & Std\\
\midrule
Caltech & 1-Perfect & 70.00 & 7.07\\
Caltech & 2-Perfect & 60.23 & 2.87\\
Caltech & 3-Perfect & 50.19 & 2.85\\
UCI & 1-Perfect & 60.94 & 2.94\\
UCI & 2-Perfect & 49.54 & 0.83\\
UCI & 3-Perfect & 42.82 & 0.92\\
UCI & 3-Partial-Incl & 44.11 & 0.98\\
UCI & 3-Partial-Excl & 43.11 & 1.33\\
UCI & 3-None & 41.11 & 0.93\\
\bottomrule
\end{tabular}
\end{table}

\begin{table}[!h]
\centering
\caption{t-statistics and p-values by Treatment and Location}
\centering
\begin{tabular}[t]{llccc}
\toprule
Location & Treatment & t-stat & \makecell{p-value \\ $H_0: \mu= 0.5$ \\ $H_1: \mu > 0.5$} & 
\makecell{p-value\\$H_0: \mu= 0.5$ \\ $H_1: \mu \ne 0.5$}\\
\midrule
Caltech & 1-Perfect & 23.300 & 0.000 & 0.000\\
Caltech & 2-Perfect & 28.500 & 0.000 & 0.000\\
Caltech & 3-Perfect & 0.627 & 0.266 & 0.532\\
UCI & 1-Perfect & 42.100 & 0.000 & 0.000\\
UCI & 2-Perfect & -5.650 & 1.000 & 0.000\\
UCI & 3-None & -101.000 & 1.000 & 0.000\\
UCI & 3-Partial-Excl & -34.700 & 1.000 & 0.000\\
UCI & 3-Partial-Incl & -54.000 & 1.000 & 0.000\\
UCI & 3-Perfect & -87.700 & 1.000 & 0.000\\
\bottomrule
\end{tabular}
\end{table}

\begin{table}[!h] 
\centering
\caption{Table for \Cref{fig:Coalitions}}\label{table:fig-3-coaliution-types}
\centering
\begin{tabular}[t]{llccccc}
\toprule
Location & Treatment  & \makecell{Grand\\ Coalition} & Dictatorial & \makecell{MWC \\ Weak}& \makecell{MWC\\ Strong} & \makecell{MWC\\ NA}\\
\midrule
Caltech & 1-Perfect & 11.36 & 2.27 & 86.36 & -- & --\\
Caltech & 2-Perfect & 12.50 & 1.14 & 76.14 & 10.23 & --\\
Caltech & 3-Perfect-Excl & 23.08 & -- & 57.69 & 19.23 & --\\
Caltech & 3-Perfect-Incl & 13.46 & 1.92 & -- & -- & 84.62\\
UCI & 1-Perfect & 30.92 & 0.66 & 67.76 & 0.66 & --\\
UCI & 2-Perfect & 25.66 & -- & 25.66 & 48.68 & --\\
UCI & 3-Perfect-Excl & 47.71 & -- & 28.44 & 23.85 & --\\
UCI & 3-Perfect-Incl & 37.25 & 1.96 & -- & -- & 60.78\\
UCI & 3-None & 64.29 & -- & -- & -- & 35.71\\
UCI & 3-Partial-Excl & 41.82 & -- & -- & -- & 58.18\\
UCI & 3-Partial-Incl & 40.00 & 0.95 & 18.10 & 40.95 & --\\
\bottomrule
\end{tabular}
\end{table}

\begin{table}[!h]
\centering
\caption{Table for \Cref{fig:egalitarian}}\label{table:fig-4-egal-coalition}
\centering
\begin{tabular}[t]{llccc}
\toprule
Location & Treatment & \makecell{Non \\Egalitarian} & \makecell{MWC\\ Egalitarian} & \makecell{Grand Coalition \\ Egalitarian}\\
\midrule
Caltech & 1-Perfect & 85.23 & 7.95 & 6.82\\
Caltech & 2-Perfect & 75.00 & 14.77 & 10.23\\
Caltech & 3-Perfect-Excl & 39.74 & 44.87 & 15.38\\
Caltech & 3-Perfect-Incl & 67.31 & 23.08 & 9.62\\
UCI & 1-Perfect & 86.18 & 7.89 & 5.92\\
UCI & 2-Perfect & 48.68 & 46.71 & 4.61\\
UCI & 3-Perfect-Excl & 17.43 & 46.79 & 35.78\\
UCI & 3-Perfect-Incl & 29.41 & 47.06 & 23.53\\
UCI & 3-Partial-Excl & 18.18 & 49.09 & 32.73\\
UCI & 3-Partial-Incl & 25.71 & 51.43 & 22.86\\
UCI & 3-None & 37.86 & 19.29 & 42.86\\
\bottomrule
\end{tabular}
\end{table}

\begin{table}[!h]
\centering
\caption{Table for \Cref{{fig:MWCegalitarian}}}\label{table:fig-5-egal-share-offers}
\centering
\begin{tabular}[t]{llcc}
\toprule
Location & Treatment & \% Egalitarian Offers & Mean of proposer share\\
\midrule
Caltech & 1-Perfect & 9.21 & 77.11\\
Caltech & 2-Perfect & 7.46 & 69.43\\
Caltech & 3-Perfect-Excl & 62.22 & 52.69\\
UCI & 1-Perfect & 11.65 & 69.78\\
UCI & 2-Perfect & 33.33 & 55.82\\
UCI & 3-Perfect-Excl & 90.32 & 51.61\\
UCI & 3-Partial-Incl & 84.21 & 50.94\\
\bottomrule
\end{tabular}
\end{table}

\clearpage

\section{Additional Tables and Graphs}\label{Appendix-AdditionalTables}

\subsection{Distribution of Offers and Acceptances}\label{Appendix-AdditionalTables1}
\Cref{fig:acceptedrejected} omitted some treatments; we show results from those treatments below.

\begin{figure}[h!]
    \centering
    \includegraphics[width=5in]{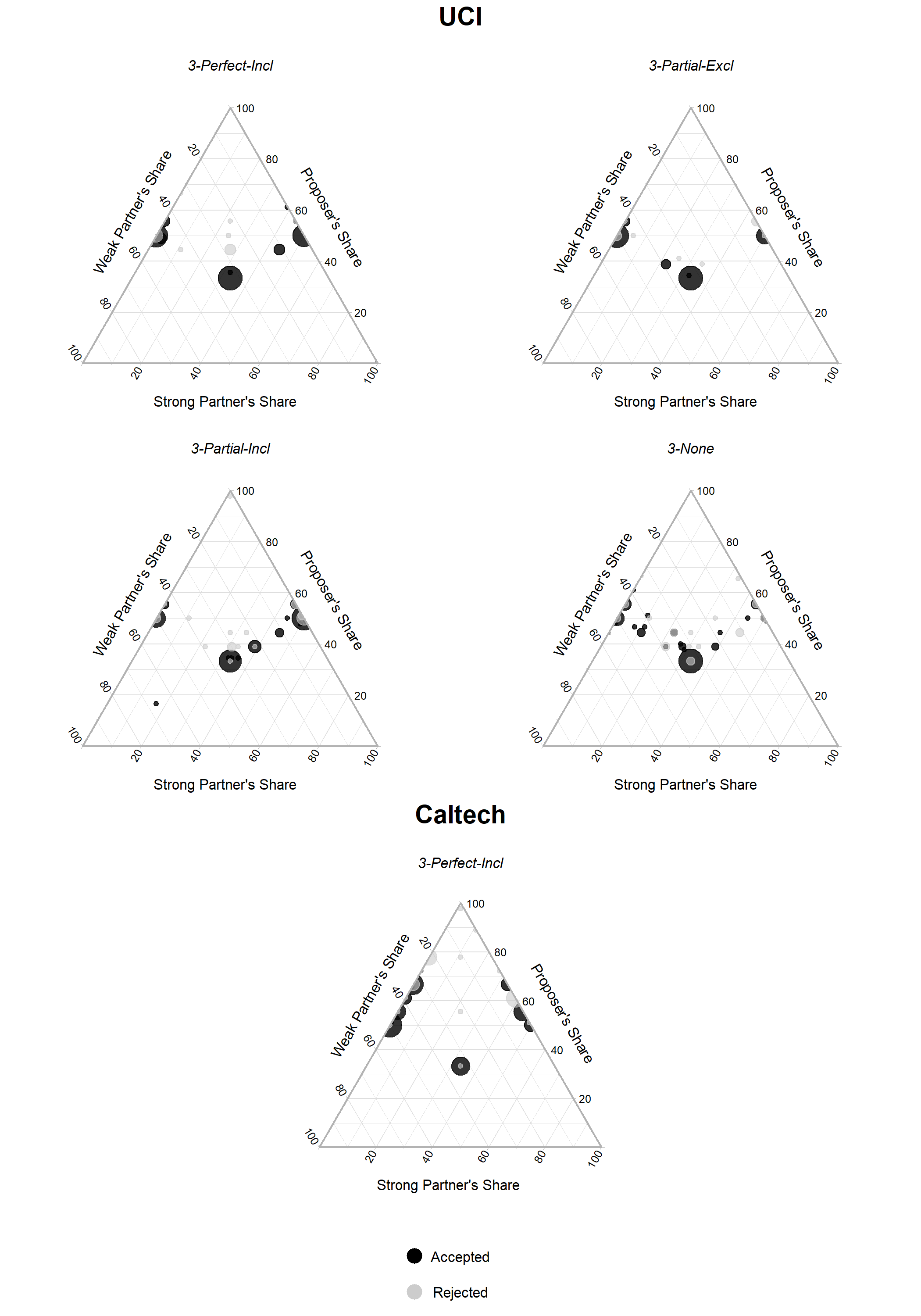}
    \caption{Distributions of Offers and Acceptances, Additional Treatments}  \label{fig:appendixacceptreject}
\end{figure}

\subsection{Mapping from First Proposals to Expected Payoffs}\label{Appendix-AdditionalTables2}
\Cref{fig:expectedpayoff} omitted some treatments; we show results from those treatments below.
\begin{figure}[h!]
    \centering
    \includegraphics[width=5in]{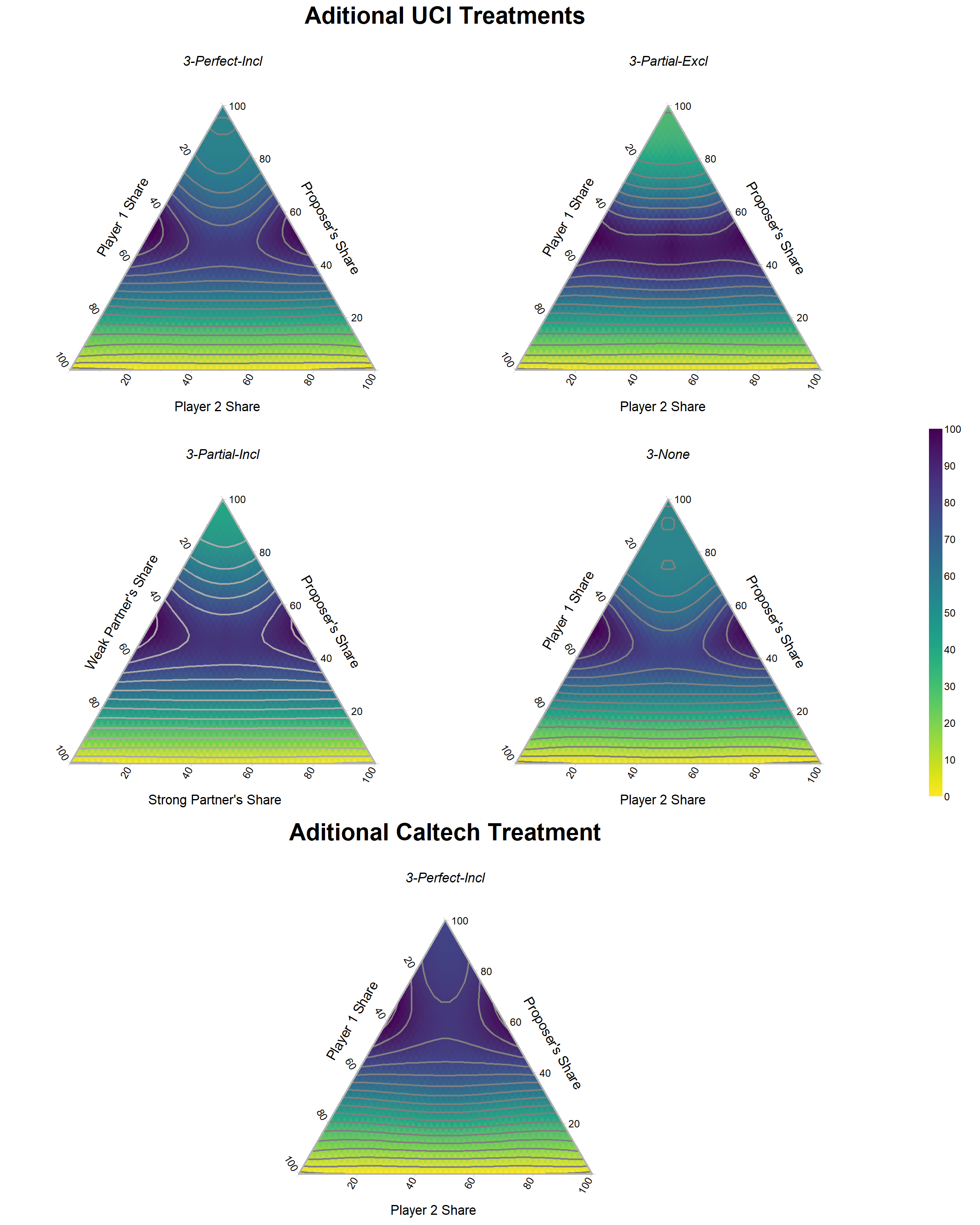}
    \caption{Mappings from First Proposals to Expected Payoffs, Additional Treatments}   \label{fig:appendixheatmap}
\end{figure}\clearpage

\subsection{Alternate Figures for \Cref{sec:OutcomesFirstProposer}}\label{Appendix:AlternateFigs}

\Cref{fig:FirstProposer,fig:Distr} used the data from experienced matches in which the first offer was accepted. Below, we produce the corresponding figures of final payoffs across \textit{all} experienced matches, \textit{including matches where the first proposal was rejected}.

\begin{figure}[h!]
    \centering
    \includegraphics[width=1\linewidth]{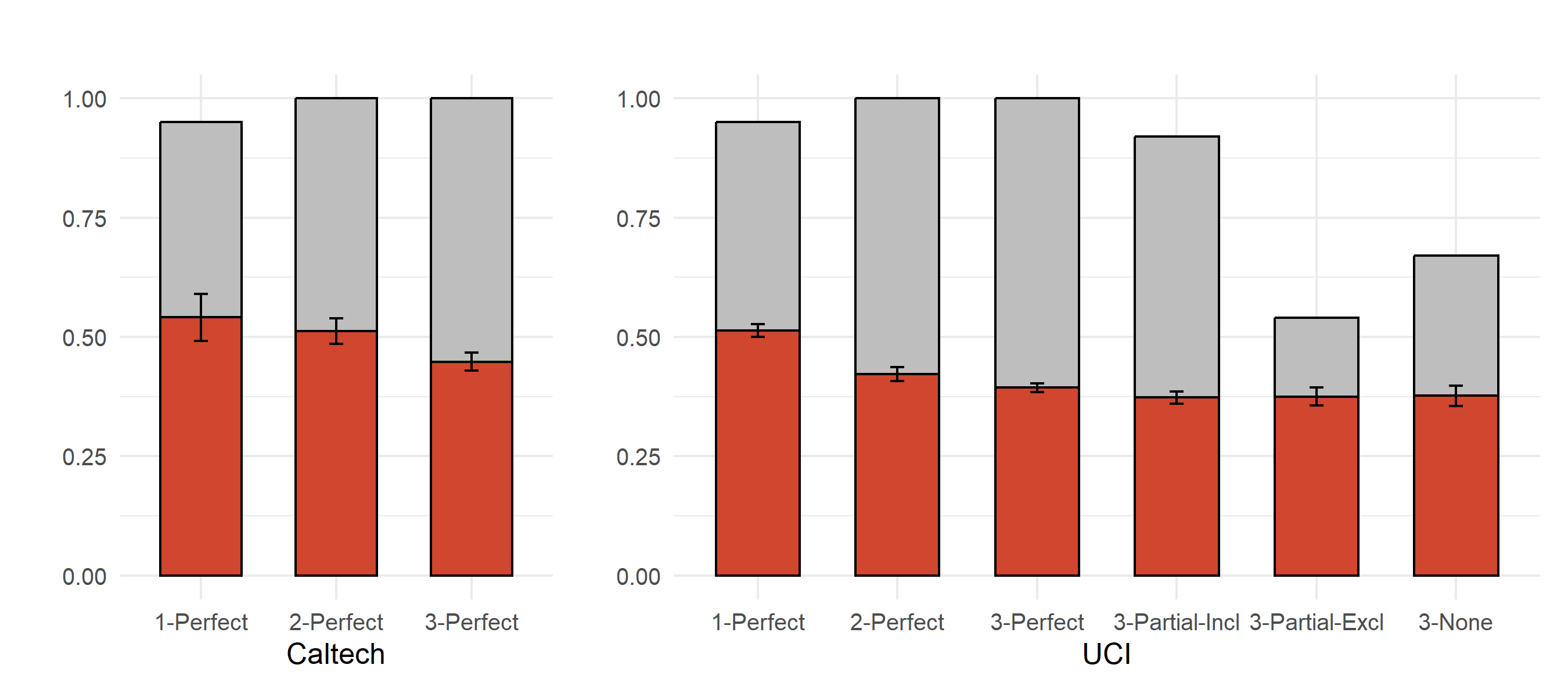}
    \caption{Realized Final Share of First Proposer Across All Experienced Matches (Including Matches Where the First Proposal was Rejected).}
    \label{fig:fig1-alt-all-offers}
\end{figure}

\begin{figure}[h!]
    \centering
    \includegraphics[width=1\linewidth]{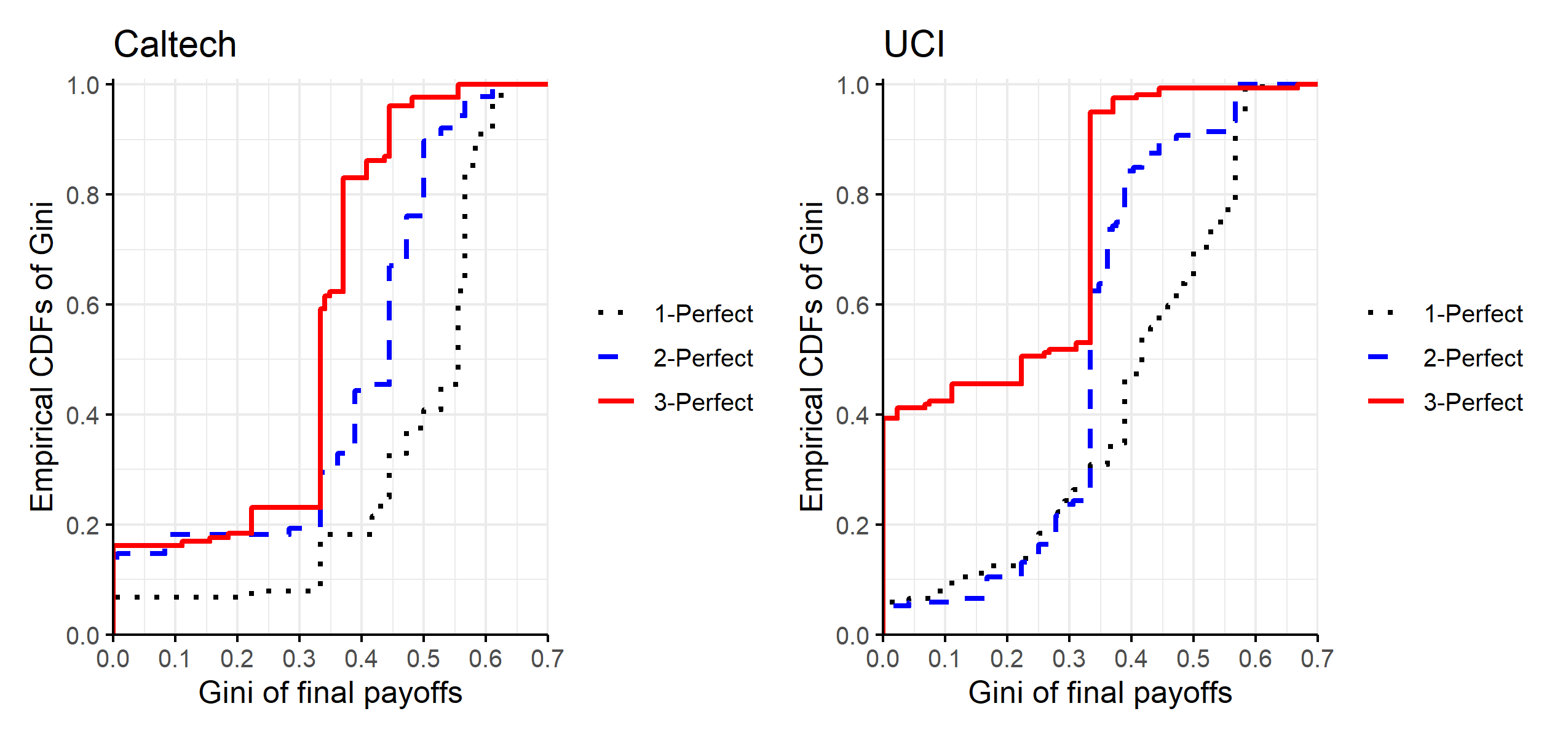}
    \caption{CDFs of Gini Coefficients of Final Payoffs Across All Experienced Matches (Including Matches Where the First Proposal was Rejected)}
    \label{fig:fig2-alt-all-offers}
\end{figure}

\clearpage

\subsection{Tables for \Cref{sec:Optimality}}\label{Appendix-Optimality}

First, we display the mean rejection payoffs for each treatment, alongside the rejection rates of the first offer. These are used in the computation of our two measures of proposer optimization.

\begin{table}[!h]\label{table:MRP}
\centering
\caption{Mean Rejection Payoffs}
\centering
\begin{tabular}[t]{llcc}
\toprule
Location & Treatment & Rejection rate & Mean Rejection Payoff\\
\midrule
Caltech & 1-Perfect & 22.73 & 5.00\\
Caltech & 2-Perfect & 27.27 & 27.24\\
Caltech & 3-Perfect-Excl & 29.49 & 26.09\\
Caltech & 3-Perfect-Incl & 38.46 & 43.00\\
UCI & 1-Perfect & 15.79 & 5.00\\
UCI & 2-Perfect & 31.58 & 26.28\\
UCI & 3-Perfect-Excl & 21.10 & 26.96\\
UCI & 3-Perfect-Incl & 21.57 & 25.76\\
UCI & 3-None & 20.71 & 24.52\\
UCI & 3-Partial-Excl & 18.18 & 12.22\\
UCI & 3-Partial-Incl & 23.81 & 15.47\\
\bottomrule
\end{tabular}
\end{table}

The table below displays what the main text refers to as the second measure, capturing the ratio (in percentages) of the actual first proposer payoff to the predicted optimal payoff, both on the aggregate (``Agg'' column) and per coalition type. Some entries exceed 100\% because the actual payoff accrued by proposers for some coalition types exceeded the predicted optimal payoffs. 

\begin{table}[h!] 
\centering
\caption{Proposer optimization rates: Measure two}\label{table:measure2}
\begin{tabular}[t]{>{}llccccccc}
\toprule
Location & Treatment & \makecell{Optimal\\Payoff}  & \makecell{MWC\\Weak} & \makecell{MWC\\Strong} & \makecell{MWC\\NA} & \makecell{Grand\\Coalition} & Dictatorial & Agg\\
\midrule
\textbf{Caltech} & 1-Perfect & 64 & 93 & -- & -- & 49 & 8 & 86\\
\textbf{Caltech} & 2-Perfect & 60 & 94 & 63 & -- & 60 & 56 & 86\\
\textbf{Caltech} & 3-Perfect-Excl & 46 & 104 & 62 & -- & 75 & -- & 89\\
\textbf{Caltech} & 3-Perfect-Incl & 55 & -- & -- & 96 & 62 & 89 & 91\\
\textbf{UCI} & 1-Perfect & 65 & 86 & 8 & -- & 66 & 140 & 80\\
\textbf{UCI} & 2-Perfect & 49 & 92 & 86 & -- & 78 & -- & 86\\
\textbf{UCI} & 3-Perfect-Excl & 44 & 97 & 97 & -- & 78 & -- & 88\\
\textbf{UCI} & 3-Perfect-Incl & 47 & -- & -- & 94 & 75 & 0 & 85\\
\textbf{UCI} & 3-None & 45 & -- & -- & 98 & 75 & -- & 83\\
\textbf{UCI} & 3-Partial-Excl & 41 & -- & -- & 104 & 76 & -- & 92\\
\textbf{UCI} & 3-Partial-Incl & 42 & 105 & 100 & -- & 71 & 0 & 88\\
\bottomrule
\end{tabular}
\end{table}

The following table presents the fraction of the optimal payoff an equal split offer to a grand coalition would produce on average for each location-treatment.

\begin{table}[!h]
\centering
\caption{Optimization rates for Equal-split Grand Coalition offers}\label{table:GCEgalitarian}
\begin{tabular}[t]{>{}llcc}
\toprule
Location & Treatment & \makecell{Optimal\\Payoff} & \makecell{GC \\ Egalitarian}\\
\midrule
\textbf{Caltech} & 1-Perfect & 64 & 50\\
\textbf{Caltech} & 2-Perfect & 60 & 55\\
\textbf{Caltech} & 3-Perfect-Excl & 46 & 70\\
\textbf{Caltech} & 3-Perfect-Incl & 55 & 66\\
\textbf{UCI} & 1-Perfect & 65 & 50\\
\textbf{UCI} & 2-Perfect & 49 & 65\\
\textbf{UCI} & 3-Perfect-Excl & 44 & 75\\
\textbf{UCI} & 3-Perfect-Incl & 47 & 70\\
\textbf{UCI} & 3-None & 45 & 72\\
\textbf{UCI} & 3-Partial-Excl & 41 & 80\\
\textbf{UCI} & 3-Partial-Incl & 42 & 75\\
\bottomrule
\end{tabular}
\end{table}

Finally, the following table presents the fraction of the optimal payoff that the first proposer obtains if her offer is rejected.

\begin{table}[!h]
\centering
\caption{Optimization rates for mean rejected payoff}\label{tb:MRPoptim}
\begin{tabular}[t]{>{}llcc}
\toprule
Location & Treatment & Optimal Payoff & MRP\\
\midrule
\textbf{Caltech} & 1-Perfect & 64 & 8\\
\textbf{Caltech} & 2-Perfect & 60 & 46\\
\textbf{Caltech} & 3-Perfect-Excl & 46 & 56\\
\textbf{Caltech} & 3-Perfect-Incl & 55 & 78\\
\textbf{UCI} & 1-Perfect & 65 & 8\\
\textbf{UCI} & 2-Perfect & 49 & 53\\
\textbf{UCI} & 3-Perfect-Excl & 44 & 61\\
\textbf{UCI} & 3-Perfect-Incl & 47 & 55\\
\textbf{UCI} & 3-None & 45 & 54\\
\textbf{UCI} & 3-Partial-Excl & 41 & 30\\
\textbf{UCI} & 3-Partial-Incl & 42 & 37\\
\bottomrule
\end{tabular}
\end{table}

\newpage


\pagebreak

\section{Results from all matches}\label{Appendix-AllMatches}

The following figures replicate Figures 1-6 using data from all matches.

\begin{figure}[h!]
    \centering
    \caption{Realized Proposer Share of Accepted First-round Offers}
    \includegraphics[width=1\linewidth]{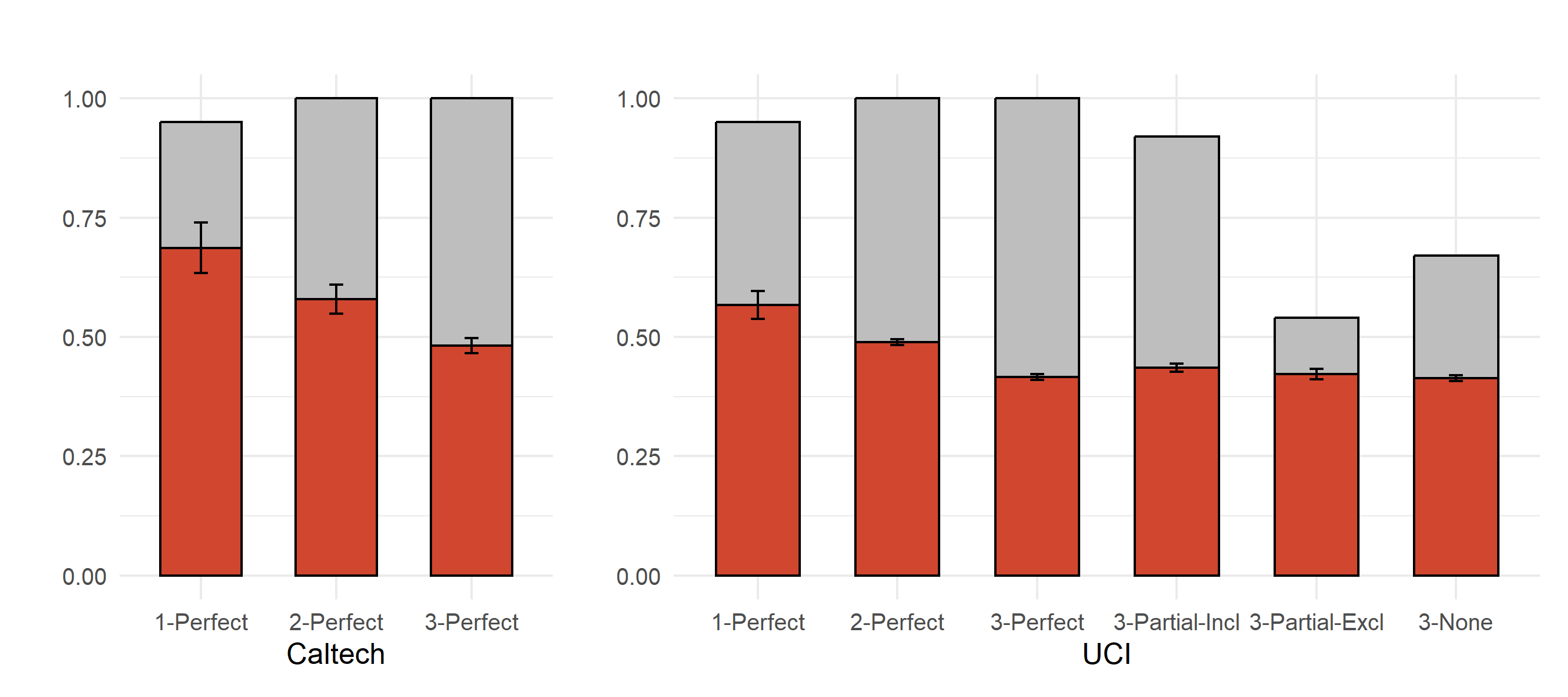}
\end{figure}

\begin{figure}[h!]
    \centering
    \caption{CDFs of Gini Coefficients of Final Payoffs of First-round Accepted Offers}
    \includegraphics[width=1\linewidth]{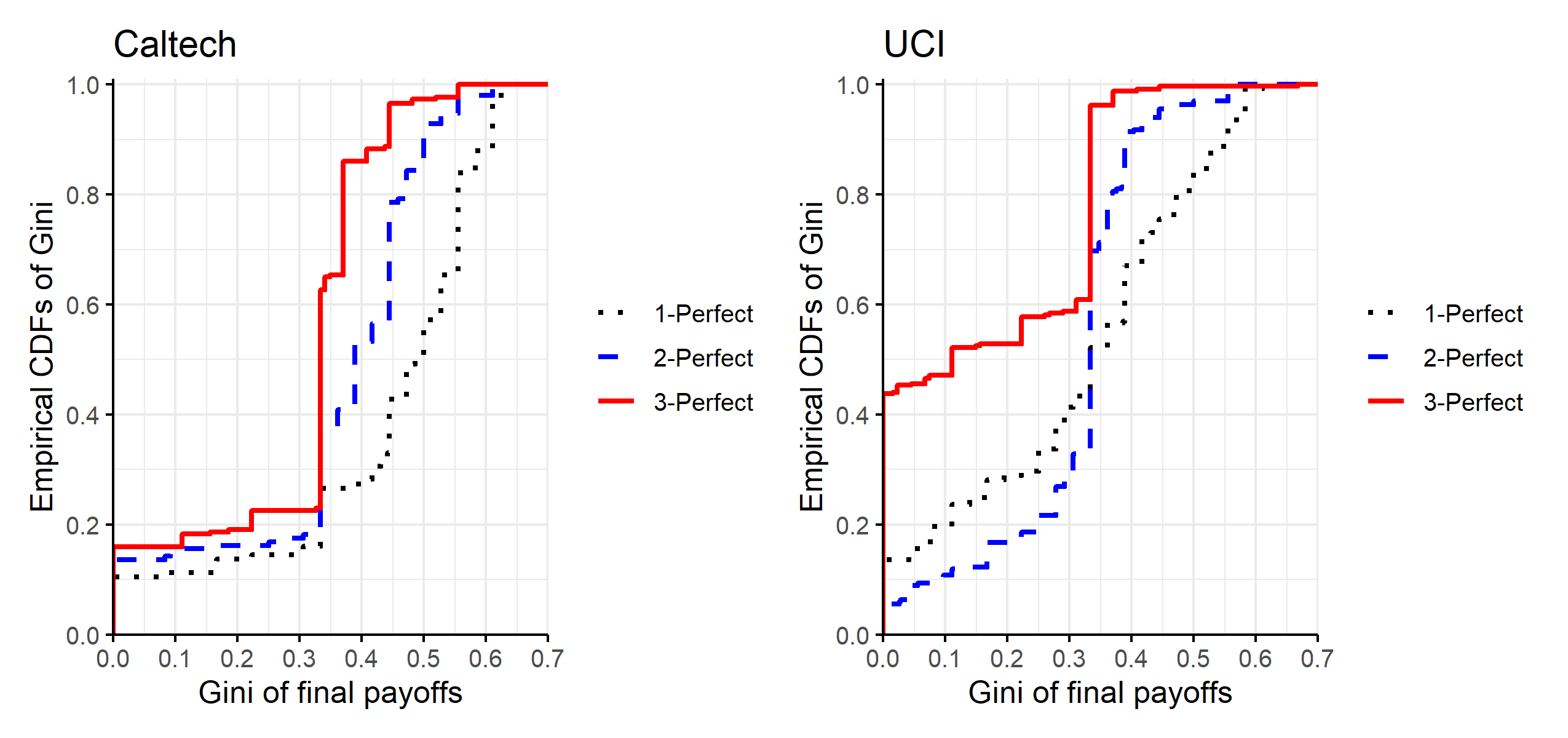}  
\end{figure}


\begin{figure}[h!]
    \centering
    \caption{Coalition types in first-round proposals}
    \includegraphics[width=5in]{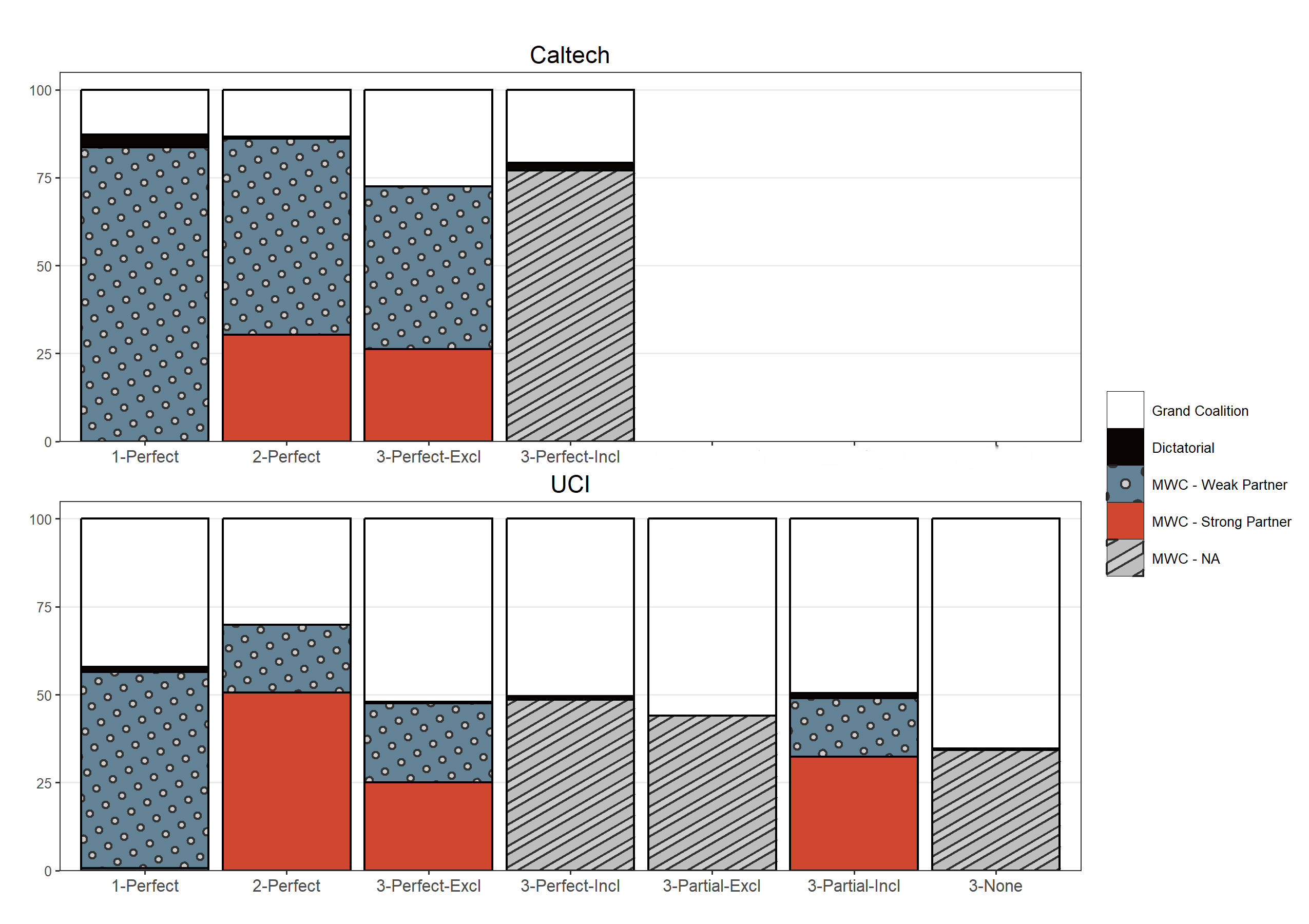}
\end{figure}


\begin{figure}[h!]
    \centering
    \caption{The Frequency of Egalitarian Offers}
    \includegraphics[width=5in]{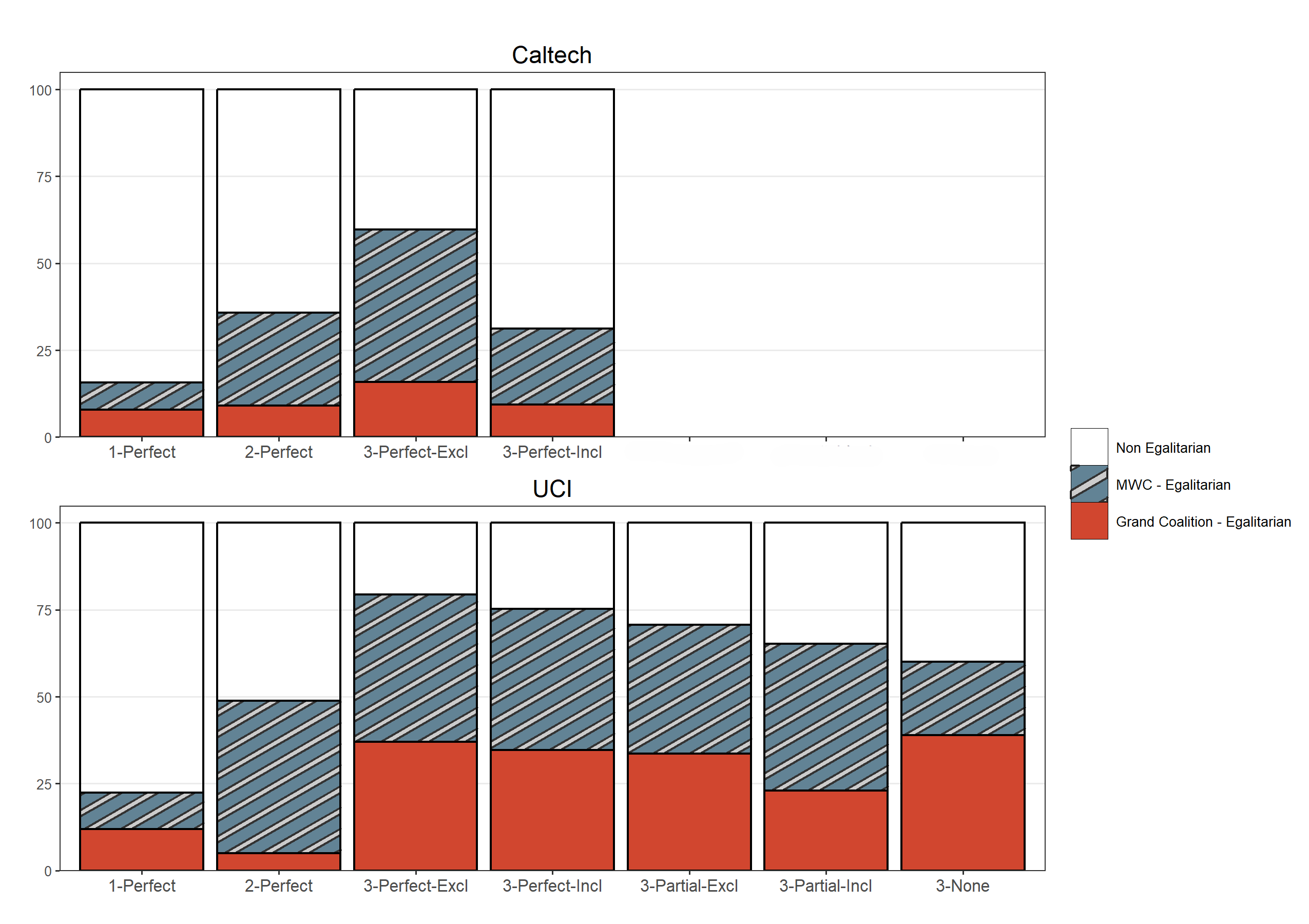}
\end{figure}


\begin{figure}[h!]
    \centering
    \caption{Frequency of Egalitarianism Among MWC Offers to Weak Partners}
    \includegraphics[width=5in]{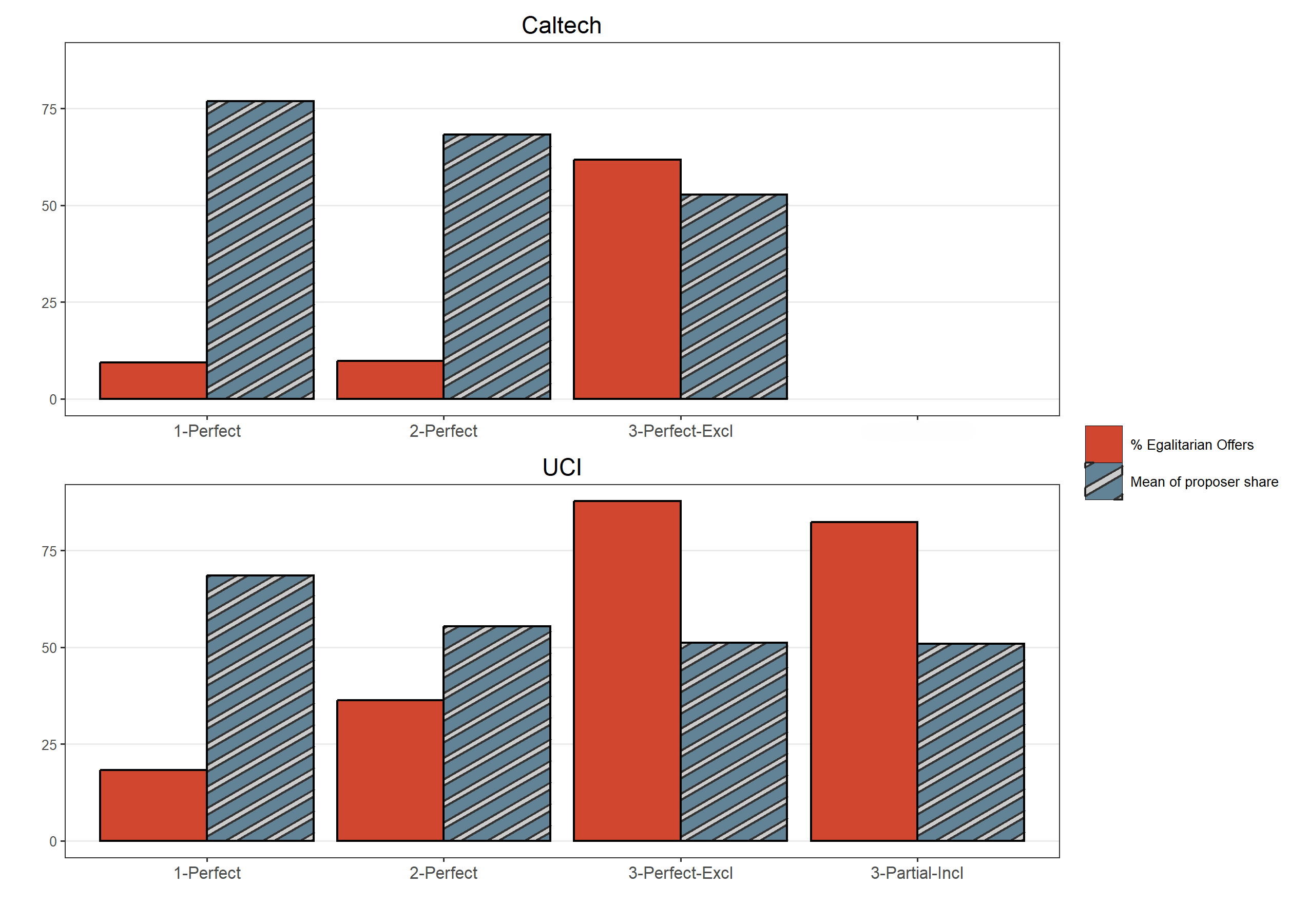}
\end{figure}

\clearpage


\begin{figure}[h]
    \centering
       \caption{Distribution of Offers and Acceptances}
    \includegraphics[width=5.7in]{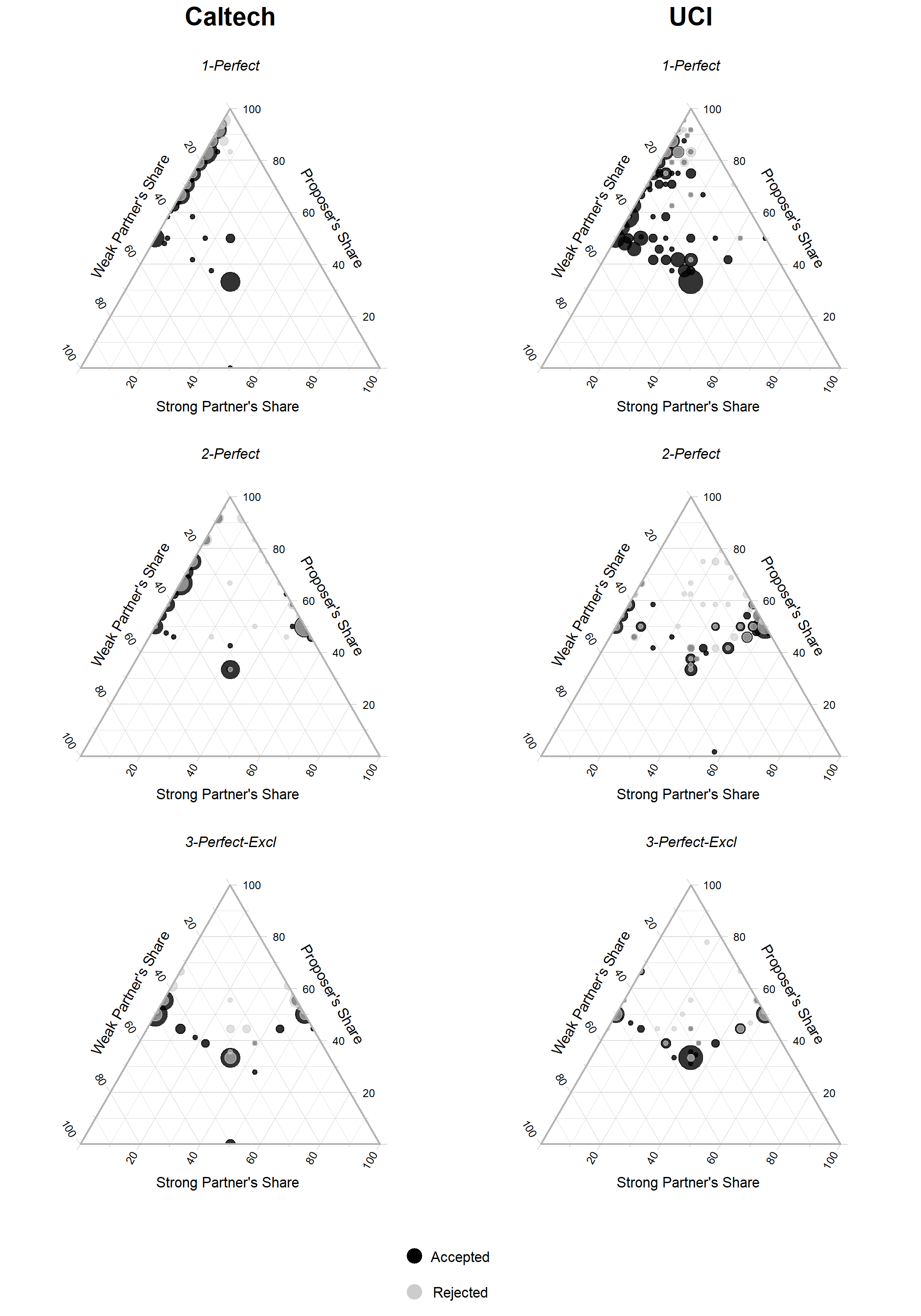}
\end{figure}

\clearpage

\section{Sample Instructions (3-Perfect treatment)}\label{Appendix-Instructions}

The following instructions were read aloud to subjects at the beginning of a 3-Perfect session:

Thank you for agreeing to participate in this group decision-making experiment. During the experiment we require your complete, undistracted attention, and ask that you follow instructions carefully.  You may not open other applications on your computer, talk with other students, or engage in other distracting activities, such as using your phone, reading books, etc.	

You will be paid for your participation in cash, at the end of the experiment.  Different participants may earn different amounts.  What you earn depends partly on your decisions, partly on the decisions of others, and partly on chance. Your earnings during the experiment are denominated in points. At the end of the experiment, the points that you earn will be converted into US dollars using the rate: 1 point = 10 cents.

The entire experiment will take place through computer terminals, and all interaction between you will take place through the computers. If you have any questions during the instruction period, raise your hand and your question will be answered out loud so everyone can hear.  If you have any questions after the experiment has begun, raise your hand, and an experimenter will come and assist you.

You will make choices over a sequence of 15 matches. At the end of the session, the computer will randomly select one match for payment. Each match is equally likely to be selected for payment. You will be paid what you have earned in this selected match, plus the show-up fee of $\$15$. Everyone will be paid in private and you are under no obligation to tell others how much you earned.

WHAT HAPPENS IN EVERY MATCH

At the beginning of every Match, all subjects are randomly divided into 3-member groups.  In addition, each member is randomly assigned an ID number, either 1, 2, or 3. The group assignment and the ID assignments are random, so you will not know the identity of the subjects you are matched with and your group-members will not know your identity. The groups are completely independent of each other and payoffs and decisions in one group have no effect on payoffs and decisions in other groups.

Each group will decide how to allocate budget among the 3 members in your group. Proposals will be voted up or down (accepted or rejected) by majority rule. That is, whenever at least 2 out of 3 members approve a proposal, it passes.

Each Match consists of potentially several Rounds. Your ID number will stay the same in all Rounds of a Match and will not be re-assigned until the next Match. At the beginning of Round 1, your group has a budget of 240 points and needs to decide how to allocate these 240 points among the 3 members in your group.  One of 3 members in your group will be randomly chosen to make a proposal of how to allocate 240 points among the 3 members. We will call this member the proposer. Each member has the same chance of being selected as the proposer. Allocations to each member must be between 0 and 240 points. All allocations must add up to 240 points.

After the selected proposer has made his/her proposal, this proposal will be posted on your computer screens with the proposed allocation to you and the other members clearly indicated. You will then have to vote either to accept or to reject this proposal.
If the proposal passes (gets 2 or more yes votes), the proposed allocation is implemented and the Match is over. If the proposal is defeated (gets 1 or less votes), there will be another Round of the same Match. However, the available budget will be reduced by 5
	
In every Round, each group member is equally likely to be selected as the proposer. However, at the beginning of each Round will be told both the ID number of the proposer in the current Round and the ID number of the proposer in the next Round should your group reach the next Round. This process will repeat itself until a proposed allocation passes (gets 2 or more yes votes).

To summarize, the steps of the process will work as follows:
1.	At the beginning of each Round, one member is randomly selected to be the proposer. The proposer makes a proposal of how to split available budget. In the very first Round, the budget is 240 points.
2.	A vote is held (each member of the group votes to accept or reject the proposal).
3.	If 2 or more members of the group accept it (vote yes), then the proposal passes and the Match is over. If the proposal is rejected, then the budget shrinks by 5

Recall, that in each Match, you will be randomly re-matched into groups of 3 members each. Each member of the group will be assigned an ID number (from 1 to 3), which is displayed on the top of the screen. Once the Match is over, you will be randomly re-matched to form new groups of 3 members and you will be assigned a NEW ID for the next Match. Please make sure you know your ID number when making your decisions. Since ID numbers will be randomly assigned prior to the start of each Match, everyone’s ID number in a Match is only temporary, and will usually change from Match to Match.

COMMUNICATION

In each Round, after the proposer is selected but before he/she submits his/her proposal, members of a group will have the opportunity to communicate with each other using the chat box. The communication is structured as follows. On the top of the screen, each member of the group will be told his/her ID number. You will also know the ID number of the member who is currently selected to make a proposal (proposer). Below you will see a box, in which you will see all messages sent to either all members of your group or to you personally. You will not see the chat messages that are sent privately to other members. In the box below that one, you can type your own message and send it either to the entire group or to particular members of your group. To select members that will receive your message, simply click on the buttons that correspond to the ID numbers of the members who you want to receive this message and hit SEND. You can send message to all members of your group by clicking the SELECT ALL button.

The chat option will be available until the proposer submits her proposal. At this moment the chat option will be disabled.
SCREEN LAYOUT

On the left side of your screen, above the chat box, there is a table, which displays information about the current Match. In this table, you will be shown information about the past Rounds of the current Match, the current round of this Match as well as information about the next Round proposer, which will be relevant if the current Round proposal fails to obtain a majority of votes.

Take a look at this table. The first row of the table refers to Round 1. Specifically, the first row of the table displays the available budget in Round 1 (240 points) as well as the ID number of the proposer. The table also displays some information about Round 2 of the game (the second row of the table) in case the Round 1 proposal fails. In particular, you can see that the budget that would be available in Round 2 is 228 points and  you  can  also  see  the  ID  number  of  the  second  round  proposer  should  the  group  reach  second round. Recall that in each Round, the proposer is chosen randomly among the three members of your group.  However, in every Round you will be told which member was selected to be the proposer in the next Round should the current Round proposal fail. The ID number of the next round proposer is displayed in the last column of the table labeled Next Round Proposers.

After the Round 1 proposer submits his/her proposal, this proposal will be displayed in the table in the column “Proposal” with your own allocation highlighted in RED. Similarly, after all members of your group submit their votes, all the votes will be listed in this table in the column “Votes” with your vote highlighted in RED.

At the beginning of Round 2, a third row will appear in this table with information about the available budget and potential proposers in Round 3 should you reach this Round. Similarly, at the beginning of Round 3, a fourth row will appear in the table, etc...

In addition, there is a history box at the bottom of the screen. This history box lists the final allocation reached by your group in each Match as well as votes of all group members.

REVIEW

1.	The experiment will consist of 15 Matches. There may be several Rounds in each Match.
2.	Prior to each Match, you will be randomly divided into groups of 3 members each. Each member in a group will be assigned an ID number.
3.	At the start of each Match, one member in your group will be randomly selected to make a proposal of how to allocate 240 points among the three of you. Before he/she submits his/her proposal, members of the group can use the chat box to communicate with each other. You may send public messages that will be delivered to all members of your group as well private messages that will be delivered to members that you specify explicitly.
4.	Proposals to each member must be greater than or equal to 0 points.
5.	If a simple majority accepts the proposal (2 or more members), the Match ends.
6.	If a simple majority rejects the proposal then the group moves on to Round 2, in which the available budget is reduced by 5
7.	Each of the three group members is equally likely to be chosen as the proposer in each Round of each Match. However, at the beginning of each Round, the group members will be told the ID number of the current Round proposer as well as the ID number of the next Round proposer should your group reach the next Round.

Are there any questions? We will now go slowly through one practice Match to familiarize you with the screen. After the practice Match is over, we will start the experiment, in which you will play 15 Matches.

\end{document}